\documentclass[journal]{IEEEtran}
\usepackage{graphicx}
\usepackage{epstopdf}
\usepackage{epsfig}
\usepackage{cite}
\usepackage{amssymb}
\usepackage{multicol}

\graphicspath{{simulation_graph/}}

\newtheorem{thm}{Theorem}
\newtheorem{col}{Corollary}

\newtheorem{rmk}{Remark}
\newtheorem{df}{Definition}
\newtheorem{exm}{Example}

\begin{document}

\title{Completely Distributed Guaranteed-performance Consensualization for High-order Multiagent Systems with Switching Topologies}
\author{{Jianxiang~Xi, Cheng~Wang, Hao~Liu$^*$, Le~Wang}\vspace{-1em}
\thanks{This work was supported by the National Natural Science Foundation of China under Grants 61374054, 61503012, 61503009, 61333011 and 61421063, and by Innovation Foundation of High-Tech Institute of Xi'an under Grants 2015ZZDJJ03, also supported by Innovation Zone Project under Grants 17-163-11-ZT-004-017-01.}
\thanks{Jianxiang~Xi,~Cheng~Wang,~Zhong~Wang are with High-Tech Institute of Xi'an, Xi'an, 710025, P.R. China, and Hao~Liu is with School of Astronautics, Beihang University, Beijing, 100191, P.R. China. (e-mail: liuhao13@buaa.edu.cn).}}

\markboth{IEEE TRANSACTIONS ON SYSTEMS, MAN, AND CYBERNETICS: SYSTEMS}{Shell
\MakeLowercase{\textit{et al.}}: Bare Demo of IEEEtran.cls for
Journals}
\maketitle

\begin{abstract}
The guaranteed-performance consensualization for high-order linear and nonlinear multiagent systems with switching topologies is respectively realized in a completely distributed manner in the sense that consensus design criteria are independent of interaction topologies and switching motions. The current paper firstly proposes an adaptive consensus protocol with guaranteed-performance constraints and switching topologies, where interaction weights among neighboring agents are adaptively adjusted and state errors among all agents can be regulated. Then, a new translation-adaptive strategy is shown to realize completely distributed guaranteed-performance consensus control and an adaptive guaranteed-performance consensualization criterion is given on the basis of the Riccati inequality. Furthermore, an approach to regulate the consensus control gain and the guaranteed-performance cost is proposed in terms of linear matrix inequalities. Moreover, main conclusions for linear multiagent systems are extended to Lipschitz nonlinear cases. Finally, two numerical examples are provided to demonstrate theoretical results.

\end{abstract}


\section{Introduction}\label{section1}
\PARstart{M}{O}{T}{I}{V}{A}{T}{E}{D} by extensive applications in different fields, such as flocking, distributed computation, synchronization, formation and containment control \cite{c1}-\hspace{-0.5pt}\cite{c10}, {\it et al}., distributed cooperative control for dynamical multiagent systems has received a great deal of attention in different engineering communities in recent years. In many practical applications of multiagent systems, all agents are required to achieve an agreement on some quantity, which is usually referred to consensus or synchronization. Generally speaking, according to the autonomous dynamics of each agent, multiagent systems can be categorized into linear ones and nonlinear ones. For linear multiagent systems, the whole dynamics can often be divided into two parts: consensus dynamics and disagreement dynamics, where disagreement dynamics is independent of consensus dynamics. However, for nonlinear multiagent systems, the whole dynamics cannot be decomposed, where the Lipschitz nonlinearity has specific connotation and has been extensively discussed. By using structure features of the above two types of multiagent systems, different consensus protocols were proposed and some important and interesting conclusions were shown in \cite{c11}-\hspace{-0.5pt}\cite{c18}, where the consensus performance constraints were not considered. \par

For many practical multiagent systems to achieve consensus, some constrains should be imposed, such as network structure constraints \cite{c19}, motion constraints of the maximum speed and acceleration \cite{c190}, and utility constraints \cite{c191}. If consensus constrains include certain cost functions, which are required to be minimum or maximum, then these problems can be modeled as optimal or suboptimal consensus. The cost functions can be divided into the individual ones and the global ones. For the individual cost function, some global goals are realized by optimizing the local objective function of each agent, as shown in \cite{c20} and \cite{c21}. For the global cost function, it is required that the whole consensus performance is minimized or maximized by local interactions among neighboring agents. For first-order multiagent systems, Cao and Ren \cite{c22} proposed a linear quadratic global cost function and gave optimal consensus criteria under the condition that the interaction topology was modeled as a complete graph. For second-order multiagent systems, Guan {\it et al}. \cite{c23} discussed guaranteed-performance consensus by the hybrid impulsive control approach. For high-order multiagent systems, guaranteed-performance consensus analysis and design problems were investigated in \cite{c24,c25,c26,c27}. It should be pointed out that guaranteed-performance consensus is intrinsically suboptimal and it is difficult to achieve optimal consensus for second-order and high-order multiagent systems with global cost functions. Furthermore, guaranteed-performance consensus criteria in \cite{c24,c25,c26,c27} are not completely distributed since they depend on the Laplacian matrix of the interaction topology or its nonzero eigenvalues. \par

In \cite{c28,c29,c30}, an interesting scaling-adaptive strategy was proposed to realize completely distributed consensus control for linear and Lipschitz nonlinear multiagent systems without performance constraints, where the impacts of the nonzero eigenvalues of the Laplacian matrix were eliminated by introducing a scaling factor. The scaling factor is inversely proportional to the minimum nonzero eigenvalue of the interaction topology and cannot be precisely determined since the minimum nonzero eigenvalue is dependent on the algebraic connectivity of the interaction topology which is difficult to be determined. When performance constraints are not involved, it is not necessary to determine the scaling factor. However, the precise value of the scaling factor is required when there are global cost functions. Therefore, the scaling-adaptive strategy cannot be applied to investigate multiagent systems with performance constraints. To the best of our knowledge, the following interesting and challenging guaranteed-performance consensus problems are still open: (i) How to achieve completely distributed guaranteed-performance consensus; (ii) How to determine the impacts of switching topologies with adaptively adjusting weights and the Lipschitz nonlinearity; (iii) How to regulate the consensus control gain and to guarantee consensus performance among all agents.\par

The current paper proposes a completely distributed guaranteed-performance consensus scheme in the sense that consensualization criteria do not depend on any information of interaction topologies and switching motions. Firstly, a new guaranteed-performance consensus protocol with switching topologies and adaptively adjusting weights is constructed, which can regulate consensus performance among all agents instead of only neighboring agents. Then, by using the specific feature of a complete graph that all its eigenvalues are identical, a translation-adaptive strategy is given to realize guaranteed-performance consensus in a completely distributed manner and adaptive guaranteed-performance consensualization criteria for high-order linear multiagent systems with switching topologies are presented. Furthermore, a regulation approach of the consensus control gain and the guaranteed-performance cost is shown by the linear matrix inequality (LMI) tool. Finally, adaptive guaranteed-performance consensus design criteria for high-order Lipschitz nonlinear multiagent systems with switching topologies are proposed. \par

Compared with closely related works on consensus, the current paper has four novel features as follows. Firstly, the current paper proposes a new translation-adaptive strategy to realize completely distributed guaranteed-performance consensus control. The methods for guaranteed-performance consensus in \cite{c23,c24,c25,c26,c27} are not completely distributed and the scaling-adaptive strategy in \cite{c28,c29,c30} cannot guarantee the consensus regulation performance. Secondly, the current paper determines the impacts of switching topologies with time-varying weights and guarantees the consensus regulation performance among all agents. In \cite{c23,c24,c25,c26,c27}, the impacts of switching topologies with time-varying weights on the consensus regulation performance among neighboring agents cannot be dealt with. Thirdly, an explicit expression of the guaranteed-performance cost is determined and an approach to regulate the consensus control gain and the guaranteed-performance cost is presented. The methods in \cite{c23,c24,c25,c26,c27} cannot determine the impacts of time-varying weights on the guaranteed-performance cost and cannot regulate the consensus control gain. Fourthly, the current paper investigates these cases that each agent contains Lipschitz nonlinear dynamics. In \cite{c23,c24,c25,c26,c27}, it was supposed that the dynamics of each agent is linear. \par

The remainder of the current paper is organized as follows. Section II models interaction topologies among agents by switching connected graphs with time-varying weights and describes the adaptive guaranteed-performance consensus problem. In Section III, adaptive guaranteed-performance consensualization criteria for high-order linear multiagent systems are presented. Section IV extends main conclusions for high-order linear multiagent systems to high-order nonlinear ones. Numerical simulations are given to illustrate theoretical results in Section V, followed by some concluding remarks in Section VI. \par

\emph{Notations:} ${\mathbb{R}^d}$ denotes the real column vector space of dimension $d$ and ${\mathbb{R}^{d \times d}}$ stands for the set of $d \times d$ dimensional real matrices. ${I_d}$ denotes the identity matrix of dimension $d$. $0$ and ${\bf{0}}$ represent the zero number and the zero column vector with a compatible dimension, respectively. ${{\bf{1}}_N}$ represents an $N$-dimensional column vector, whose entries are equal to $1$. ${Q^T}$ and ${Q^{ - 1}}$ denote the transpose and the inverse matrix of $Q$, respectively. ${R^T} = R > {\rm{0}}$ and ${R^T} = R < {\rm{0}}$ mean that the symmetric matrix $R$ is positive definite and negative definite, respectively. The notation $ \otimes $ stands for the Kronecker product. $\left\| x \right\|$ represents the two norm of the vector $x$. The symmetric terms of a symmetric matrix are denoted by the symbol *.

\section{Problem description}\label{section2}

\subsection{Modeling interaction topology}
A connected undirected graph $G$ with $N$ nodes can be used to depict the interaction topology of a multiagent system with $N$ identical agents, where each node represents an agent, the edge between two nodes denotes the interaction channel between them and the edge weight stands for the interaction strength. The graph $G$ is said to be connected if there at least exists an undirected path between any two nodes. The graph $G$ is said to be a complete graph if there exists an undirected edge between any two nodes. It is clear that a complete graph is connected. More basic concepts and results about graph theory can be found in \cite{c300}. \par

Let $\sigma (t):\left[ {\left. {0, + \infty } \right)} \right. \to \kappa $ denote the switching signal with $\kappa $ an index of the switching set consisting of several connected undirected graphs, where switching movements satisfy that ${t_m} - {t_{m - 1}} \ge {T_{\rm{d}}}~{\rm{ }}(\forall m \ge 1)$ with ${T_{\rm{d}}} > 0$ for switching sequences $\{ {t_i}:i = 0,1,2, \cdots \} $. The index set of all neighbors of node ${v_i}$ is denoted by ${N_{\sigma (t),i}} = \left\{ {k|\left( {{v_k},{v_i}} \right) \in E\left( {{G_{\sigma (t)}}} \right)} \right\}$, where $({v_k},{v_i})$ denotes the edge between node ${v_k}$ and node ${v_i}$ and $E\left( {{G_{\sigma (t)}}} \right)$ is the edge set of the graph ${G_{\sigma (t)}}$. Define the 0-1 Laplacian matrix of ${G_{\sigma (t)}}$ as ${L_{\sigma (t),0}} = \left[ {{l_{\sigma (t),ik}}} \right] \in {\mathbb{R}^{N \times N}}$ with ${l_{\sigma (t),ii}} =  - \sum\nolimits_{k = 1,k \ne i}^N {{l_{\sigma (t),ik}}} $, ${l_{\sigma (t),ik}} =  - 1$ if $({v_k},{v_i}) \in E\left( {{G_{\sigma (t)}}} \right)$ and ${l_{\sigma (t),ik}} = 0$ otherwise and the Laplacian matrix of ${G_{\sigma (t)}}$ as ${L_{\sigma (t),w}} = \left[ {{{\tilde l}_{\sigma (t),ik}}(t)} \right] \in {\mathbb{R}^{N \times N}}$ with ${\tilde l_{\sigma (t),ii}}(t) =  - \sum\nolimits_{k = 1,k \ne i}^N {{l_{\sigma (t),ik}}{w_{\sigma (t),ik}}} (t)$ and ${\tilde l_{\sigma (t),ik}}(t) = {l_{\sigma (t),ik}}{w_{\sigma (t),ik}}(t){\rm{ }}(i \ne k)$, where the function ${w_{\sigma (t),ik}}(t) \ge 1$ is designed later. It can be found that ${\tilde l_{\sigma (t),ii}}(t) = \sum\nolimits_{k \in {N_{\sigma (t),i}}} {{w_{\sigma (t),ik}}(t)} $, ${L_{\sigma (t),0}}{{\bf{1}}_N} = {L_{\sigma (t),w}}{{\bf{1}}_N} = {\bf{0}}$, and ${L_{\sigma (t),0}}$ is piecewise continuous and is constant at no switching time, but ${L_{\sigma (t),w}}$ may be not. Specially, for ${L_{\sigma (t),0}}$ and ${L_{\sigma (t),w}}$, the zero eigenvalue is simple and all the other $N - 1$ eigenvalues are positive since all the topologies in the switching set are connected. \par

\subsection{Describing guaranteed-performance consensualization}
The dynamics of each agent is described by
\begin{eqnarray}\label{1}
{\dot x_i}(t) = A{x_i}(t) + B{u_i}(t){\rm{ }}\left( {i = 1,2, \cdots ,N} \right),
\end{eqnarray}
where $A \in {\mathbb{R}^{d \times d}},$ $B \in {\mathbb{R}^{d \times p}},$ and ${x_i}(t)$ and ${u_i}(t)$ are the state and the control input of agent $i$, respectively. The following adaptive guaranteed-performance consensus protocol is proposed for agent $i$ to apply the state information of its neighbors
\begin{eqnarray}\label{2}
\left\{ \begin{array}{l}
 {u_i}(t) = {K_u}\sum\limits_{k \in {N_{\sigma (t),i}}} {{w_{\sigma (t),ik}}(t)\left( {{x_k}(t) - {x_i}(t)} \right)} , \\
 {{\dot w}_{\sigma (t),ik}}(t) = {\left( {{x_k}(t) - {x_i}(t)} \right)^T}{K_w}\left( {{x_k}(t) - {x_i}(t)} \right), \\
 {J_x} \! \! = \! \! \frac{1}{N} \! \! \sum\limits_{i = 1}^N \! \! {\sum\limits_{k = 1}^N \! \! {\int_0^{ + \infty } \! \! {{{\left( {{x_k}(t) \! - \! {x_i}(t)} \right)}^T}\! Q \! \left( {{x_k}(t) \! - \! {x_i}(t)} \right)\! {\rm{d}}t} } } , \\
 \end{array} \right.
\end{eqnarray}
where ${K_u} \in {\mathbb{R}^{p \times d}}$ and ${K_w} \in {\mathbb{R}^{d \times d}}$ are gain matrices, and ${Q^T} = Q > 0$ is used to guarantee the consensus regulation performance. It is assumed that ${w_{\sigma (t),ik}}(t)$ is a bounded function with an upper bound denoting ${\gamma _{ik}}$, ${w_{\sigma (0),ik}}(0) = 1{\rm{ }}(i \ne k)$ and ${w_{\sigma ({t_m}),ik}}({t_m}) = 1$ if the edge $({v_k},{v_i})$ is newly added at switching time ${t_m}$. In this case, ${w_{\sigma (t),ik}}(t)$ is a practical interaction strength of the channel $({v_k},{v_i})$ if agent $k$ is a neighbor of agent $i$, and can be regarded as a virtual interaction strength of the channel $({v_k},{v_i})$ if agent $k$ is not a neighbor of agent $i$. Especially, the initial value of the interaction strength is designed as $1$ once a virtual channel becomes a practical one.\par
The definition of the adaptive guaranteed-performance consensualization of high-order multiagent systems is given as follows.
\begin{df}\label{definition1}
Multiagent system (\ref{1}) is said to be adaptively guaranteed-performance consensualizable by protocol (\ref{2}) if there exist ${K_u}$ and ${K_w}$ such that ${\lim _{t \to  + \infty }}\left( {{x_i}(t) - {x_k}(t)} \right) = {\bf{0}}{\rm{ }}\left( {i,k = 1,2, \cdots ,N} \right)$ and ${J_x} \leqslant {J^*}$ for any bounded initial states ${x_i}(0){\rm{ }}\left( {i = 1,2, \cdots ,N} \right)$, where ${J^*}$ is said to be the guaranteed-performance cost.
\end{df}

The current paper mainly investigates the following four problems: (i) How to design ${K_u}$ and ${K_w}$ such that multiagent system (\ref{1}) achieves adaptive guaranteed-performance consensus; (ii) How to determine the impacts of switching topologies on adaptive guaranteed-performance consensus under the condition that interaction strengths are time-varying; (iii) How to regulate the consensus control gain and the guaranteed-performance cost; (iv) How to extend main results for high-order linear multiagent systems to high-order Lipschitz nonlinear ones.\par

\begin{rmk}\label{remark1}
The consensus protocol given in (\ref{2}) can realize a completely distributed guaranteed-performance consensus control by adaptively regulating interaction weights, but the global information of the interaction topology is required if the consensus protocol is not adaptive. Moreover, protocol (\ref{2}) has two critical characteristics. The first one is that interaction strengths ${w_{\sigma (t),ik}}(t){\rm{ (}}i \ne k{\rm{)}}$ are time-varying, while it was usually assumed that interaction strengths are time-invariant, but the neighbor set is time-varying in most consensus works with switching topologies (see \cite{c10}, \cite{c14}, \cite{c24} and references therein). Moreover, for adaptive consensus protocols with fixed topologies in \cite{c28,c29,c30}, interaction strengths may be monotonously increasing. However, for adaptive consensus with switching topologies, interaction strengths may be suddenly decreasing at some switching times. Hence, it is more difficult to determine the impacts of switching topologies on the adaptive consensus property and the upper bound of the guaranteed-performance cost. The second one is that protocol (\ref{2}) can guarantee the consensus regulation performance between any two agents by ${J_x}$ even if they are not neighboring. However, the consensus regulation performance among neighboring agents can only be ensured by the index function in \cite{c23,c24,c25,c26,c27}. Furthermore, ${J_r}(t)$ is usually called the performance regulation term, which can be realized by choosing a proper $Q$. The matrix $Q$ can be applied to ensure the regulation performance of relative motions among neighboring agents. For practical multiagent systems, $Q$ is often chosen as a diagonal matrix. In this case, a bigger coupling weight in $Q$ can ensure a smaller squared sum of the corresponding element of the state error.

\end{rmk}

\section{Adaptive guaranteed-performance consensualization criteria}\label{section3}
In this section, by the nonsingular transformation, the consensus and disagreement dynamics of multiagent system (\ref{1}) are first determined, respectively. Then, based on the Riccati inequality, adaptive guaranteed-performance consensualization criteria are proposed, and the guaranteed-performance cost ${J^*}$ is meanwhile determined. Finally, an approach to regulate the consensus control gain and the guaranteed-performance cost is presented.\par

Let $x(t) = {\left[ {x_1^T(t),x_2^T(t), \cdots ,x_N^T(t)} \right]^T}{\rm{ ,}}$ then the dynamics of multiagent system (\ref{1}) with protocol (\ref{2}) can be written as
\begin{eqnarray}\label{3}
\dot x(t) = \left( {{I_N} \otimes A - {L_{\sigma (t),w}} \otimes B{K_u}} \right)x(t).
\end{eqnarray}
Let $0 = {\lambda _{\sigma (t),1}} < {\lambda _{\sigma (t),2}} \le  \cdots  \le {\lambda _{\sigma (t),N}}$ denote the eigenvalues of ${L_{\sigma (t),0}}$, then there exists an orthonormal matrix ${U_{\sigma (t)}} = \left[ {{{{{\bf{1}}_N}} \mathord{\left/ {\vphantom {{{{\bf{1}}_N}} {\sqrt N }}} \right. \kern-\nulldelimiterspace} {\sqrt N }},{{\tilde U}_{\sigma (t)}}} \right]$ with ${\tilde U_{\sigma (t)}} \in {\mathbb{R}^{N \times (N - 1)}}$ such that
\[
U_{\sigma (t)}^T{L_{\sigma (t),0}}{U_{\sigma (t)}} = {\rm{diag}}\left\{ {{\lambda _{\sigma (t),1}},{\lambda _{\sigma (t),2}}, \cdots ,{\lambda _{\sigma (t),N}}} \right\}.
\]
Since all interaction topologies in the switching set are undirected, one has ${L_{\sigma (t),w}}{{\bf{1}}_N} = {\bf{0}}$ and ${\bf{1}}_N^T{L_{\sigma (t),w}} = {\bf{0}}$. Thus, one can show that
\[
U_{\sigma (t)}^T{L_{\sigma (t),w}}{U_{\sigma (t)}} = \left[ {\begin{array}{*{20}{c}}
   0 & {{{\bf{0}}^T}}  \\
   {\bf{0}} & {\tilde U_{\sigma (t)}^T{L_{\sigma (t),w}}{{\tilde U}_{\sigma (t)}}}  \\
\end{array}} \right].
\]
Due to ${T_{\rm{d}}} > 0$, the matrix ${U_{\sigma (t)}}$ is piecewise continuous and is constant in the switching interval. Hence, let $\tilde x(t) = \left( {U_{\sigma (t)}^T \otimes {I_d}} \right)x(t) = {\left[ {\tilde x_1^T(t),{\zeta ^T}(t)} \right]^T}$ with $\zeta (t) = {\left[ {\tilde x_2^T(t),\tilde x_3^T(t), \cdots ,\tilde x_N^T(t)} \right]^T}$, then multiagent system (\ref{3}) can be transformed into
\begin{eqnarray}\label{4}
{\dot {\tilde x}_1}(t) = A{\tilde x_1}(t),
\end{eqnarray}
\begin{eqnarray}\label{5}
\dot \zeta (t) = \left( {{I_{N - 1}} \otimes A - \tilde U_{\sigma (t)}^T{L_{\sigma (t),w}}{{\tilde U}_{\sigma (t)}} \otimes B{K_u}} \right)\zeta (t).
\end{eqnarray}
Define
\begin{eqnarray}\label{6}
{x_{\bar c}}(t) \buildrel \Delta \over = \sum\limits_{i = 2}^N {{U_{\sigma (t)}}{e_i} \otimes {{\tilde x}_i}(t)},
\end{eqnarray}
\begin{eqnarray}\label{7}
{x_c}(t) \buildrel \Delta \over = {U_{\sigma (t)}}{e_1} \otimes {\tilde x_1}(t) = \frac{1}{{\sqrt N }}{{\bf{1}}_N} \otimes {\tilde x_1}(t),
\end{eqnarray}
where ${e_i}$ $(i \in \{ 1,2, \cdots ,N\} )$ denotes an $N$-dimensional column vector with the $i$th element $1$ and $0$ elsewhere. Due to
\[
\sum\limits_{i = 2}^N {{e_i} \otimes {{\tilde x}_i}(t)}  = {\left[ {{{\bf{0}}^T},{\zeta ^T}(t)} \right]^T},
\]
one can show by (\ref{6}) that
\begin{eqnarray}\label{8}
{x_{\bar c}}(t) = \left( {{U_{\sigma (t)}} \otimes {I_d}} \right){\left[ {{{\bf{0}}^T},{\zeta ^T}(t)} \right]^T}.
\end{eqnarray}
By (\ref{7}), one has
\begin{eqnarray}\label{9}
{x_c}(t) = \left( {{U_{\sigma (t)}} \otimes {I_d}} \right){\left[ {\tilde x_1^T(t),{{\bf{0}}^T}} \right]^T}.
\end{eqnarray}
From (\ref{8}) and (\ref{9}), one can see that ${x_{\bar c}}(t){\rm{ }}$ and ${x_c}(t){\rm{ }}$ are linearly independent since ${U_{\sigma (t)}} \otimes {I_d}$ is nonsingular. Due to $\left( {U_{\sigma (t)}^T \otimes {I_d}} \right)x(t) = {\left[ {\tilde x_1^T(t),{\zeta ^T}(t)} \right]^T}$, one has \[x(t) = {x_{\bar c}}(t) + {x_c}(t){\rm{.}}\]
According to the structure of ${x_c}(t){\rm{ }}$ shown in (\ref{7}), multiagent system (\ref{1}) achieves consensus if and only if ${\lim _{t \to  + \infty }}\zeta (t) = {\bf{0}}$; that is, subsystems (\ref{4}) and (\ref{5}) describe the consensus and disagreement dynamics of multiagent system (\ref{1}).\par

Based on the above analysis, the following theorem gives a sufficient condition for adaptive guaranteed-performance consensualization in terms of the Riccati inequality, which can realize completely distributed guaranteed-performance consensus control.\par

\begin{thm}\label{theorem1}
For any given translation factor $\gamma  > 0$, multiagent system (\ref{1}) is adaptively guaranteed-performance consensualizable by protocol (\ref{2}) if there exists a matrix ${R^T} = R > 0$ such that
\[
RA + {A^T}R - \gamma RB{B^T}R + 2Q \le 0.
\]
In this case, ${K_u} = {B^T}R$, ${K_w} = RB{B^T}R$ and the guaranteed-performance cost satisfies that ${J^ * } = J_{x(0)}^* + J_{x(t)}^*,$ where
\[
J_{x(0)}^* = {x^T}(0)\left( {\left( {{I_N} - \frac{1}{N}{{\bf{1}}_N}{\bf{1}}_N^T} \right) \otimes R} \right)x(0),
\]
\[
J_{x(t)}^* \! = \! \gamma \! \! \int_0^{ + \infty } \! {{x^T}(t) \! \left( {\left( {{I_N} - \frac{1}{N}{{\bf{1}}_N}{\bf{1}}_N^T} \right) \otimes RB{B^T}R} \right) \! x(t)} {\rm{d}}t.
\]\vspace{0.05em}
\end{thm}

\begin{proof}
First of all, we design ${K_u}$ and ${K_w}$ such that ${\lim _{t \to  + \infty }}\zeta (t) = {\bf{0}}$. Construct a new Lyapunov function candidate as follows
\[
\hspace{-7em}V(t) = {\zeta ^T}(t)\left( {{I_{N - 1}} \otimes R} \right)\zeta (t) \vspace{-4pt}
\]
\[
\hspace{3em} + \sum\limits_{i = 1}^N {\sum\limits_{k \in {N_{\sigma (t),i}}} {\frac{{{{\left( {{w_{\sigma (t),ik}}(t) + {l_{\sigma (t),ik}}} \right)}^2}}}{2}} }  \vspace{-10pt}
\]
\begin{eqnarray}\label{10}
\hspace{4em} + \frac{\gamma }{N}\sum\limits_{i = 1}^N {\sum\limits_{k = 1,k \ne i}^N {\left( {{\gamma _{ik}} - {w_{\sigma (t),ik}}(t)} \right)} } ,
\end{eqnarray}
where $\gamma  > 0$ and $R$ is the solution of $RA + {A^T}R - \gamma RB{B^T}R + 2Q \le 0$. Due to ${R^T} = R > 0$ and ${\gamma _{ik}} \ge {w_{\sigma (t),ik}}(t)$ $\left( {i,k = 1,2, \cdots ,N} \right)$, one can find that $V(t) \ge 0$. Since ${l_{\sigma (t),ik}}$ is piecewise continuous and is constant in the switching interval, the time derivative of $V(t)$ is
\[
\hspace{-1.25em}\dot V(t) = {\zeta ^T}(t)( {{I_{N - 1}} \otimes \left( {RA + {A^T}R} \right) - \tilde U_{\sigma (t)}^T{L_{\sigma (t),w}}{{\tilde U}_{\sigma (t)}} }
\vspace{-4pt}
\]
\[
\hspace{3.25em} \otimes\left. {\left( {RB{K_u}} \right.\left. { + K_u^T{B^T}R} \right)} \right)\zeta (t)
\sum\limits_{i = 1}^N {\sum\limits_{k \in {N_{\sigma (t),i}}} \hspace{-3pt}{\left( {{w_{\sigma (t),ik}}(t)} \right.} }
\vspace{-12pt}
\]
\begin{eqnarray}\label{11}
\hspace{3em}\left. { + {l_{\sigma (t),ik}}} \right){{\dot w}_{\sigma (t),ik}}(t)
\hspace{-3pt}-\hspace{-3pt} \frac{\gamma }{N}\sum\limits_{i = 1}^N {\sum\limits_{k = 1,k \ne i}^N {{{\dot w}_{\sigma (t),ik}}(t)} } .
\end{eqnarray}
From (\ref{2}), one can obtain that
\[
\sum\limits_{i = 1}^N {\sum\limits_{k \in {N_{\sigma (t),i}}}\hspace{-10pt} {\left( {{w_{\sigma (t),ik}}(t) \hspace{-3pt}+\hspace{-3pt} {l_{\sigma (t),ik}}} \right)} {{\dot w}_{\sigma (t),ik}}(t)}  \hspace{-1pt}-\hspace{-1pt} \frac{\gamma }{N}\sum\limits_{i = 1}^N \hspace{-3pt}{\sum\limits_{\scriptstyle k = 1 \hfill \atop
  \scriptstyle k \ne i \hfill}^N\hspace{-3pt} {{{\dot w}_{\sigma (t),ik}}(t)} } \vspace{-12pt}
\]
\begin{eqnarray}\label{12}
 = 2{x^T}(t)\left( {\left( {{L_{\sigma (t),w}} - {L_{\sigma (t),0}} - \gamma {L_N}} \right) \otimes {K_w}} \right)x(t),
\end{eqnarray}
where ${L_N}$ is the Laplacian matrix of a complete graph with the weights of all the edges ${1 \mathord{\left/ {\vphantom {1 N}} \right. \kern-\nulldelimiterspace} N}$. Due to ${U_{\sigma (t)}}U_{\sigma (t)}^T = {I_N}$, one can show that
\[{\tilde U_{\sigma (t)}}\tilde U_{\sigma (t)}^T = {I_N} - \frac{1}{N}{{\bf{1}}_N}{\bf{1}}_N^T = {L_N}.\]
Thus, one can derive that
 \[
\hspace{-6em} {x^T}(t)\left( {\left( {{L_{\sigma (t),w}} - {L_{\sigma (t),0}} - \gamma {L_N}} \right) \otimes {K_w}} \right)x(t) \vspace{-4pt}
 \]
 \[
= {\zeta ^T}(t)((\tilde U_{\sigma (t)}^T{L_{\sigma (t),w}}{{\tilde U}_{\sigma (t)}} - \tilde U_{\sigma (t)}^T{L_{\sigma (t),0}}{{\tilde U}_{\sigma (t)}} \vspace{-10pt}
 \]
\begin{eqnarray}\label{13}
 \hspace{-7.5em}~~- \gamma {I_{N - 1}}) \otimes {K_w})\zeta (t).
\end{eqnarray}
Let ${K_u} = {B^T}R$ and ${K_w} = RB{B^T}R$, then from (\ref{11}) to (\ref{13}), by $\tilde U_{\sigma (t)}^T{L_{\sigma (t),0}}{\tilde U_{\sigma (t)}} = {\rm{diag}}\left\{ {{\lambda _{\sigma (t),2}},{\lambda _{\sigma (t),3}}, \cdots ,{\lambda _{\sigma (t),N}}} \right\}$, one has
\[\dot V(t) = \! \sum\limits_{i = 2}^N \! {\tilde x_i^T(t) \! \left( {RA \! + \! {A^T}R \! - \! 2\left( {{\lambda _{\sigma (t),i}} \! + \! \gamma } \right)RB{B^T}R} \right) \! {{\tilde x}_i}(t)} .\]
Due to $\gamma  > 0$ and ${\lambda _{\sigma (t),i}} > 0{\rm{ }}\left( {i = 2,3, \cdots ,N} \right)$, one has
\[
\hspace{-9em}RA + {A^T}R - 2\left( {{\lambda _{\sigma (t),i}} + \gamma } \right)RB{B^T}R \vspace{-2pt}
\]
\[
\hspace{4em}\le RA + {A^T}R - 2\gamma RB{B^T}R{\rm{ }}(i = 2,3, \cdots ,N).
\]
If $RA + {A^T}R - \gamma RB{B^T}R + 2Q \le 0$, then $RA + {A^T}R - 2\gamma RB{B^T}R < 0$ since ${Q^T} = Q > 0$ and $RB{B^T}R \ge 0$. Thus, one can obtain that $\dot V(t) \le 0$ and $\dot V(t) \equiv 0$ if and only if ${\lim _{t \to  + \infty }}\zeta (t) = {\bf{0}}$, which means that multiagent system (\ref{1}) achieves adaptive consensus. \par

In the following, the guaranteed-performance cost is determined. One can show that
\[
\hspace{-6em}\frac{1}{N}\sum\limits_{i = 1}^N {\sum\limits_{k = 1}^N {{{\left( {{x_k}(t) - {x_i}(t)} \right)}^T}Q\left( {{x_k}(t) - {x_i}(t)} \right)} } \vspace{-10pt}
\]
\begin{eqnarray}\label{14}
\hspace{10em}= {x^T}(t)\left( {2{L_N} \otimes Q} \right)x(t).
\end{eqnarray}
Due to $U_{\sigma (t)}^T{L_N}{U_{\sigma (t)}} = {\rm{diag}}\left\{ {0,{I_{N - 1}}} \right\}$, one has
\begin{eqnarray}\label{15}
{x^T}(t)\left( {{L_N} \otimes Q} \right)x(t) = \sum\limits_{i = 2}^N {\tilde x_i^T(t)Q{{\tilde x}_i}(t)} .
\end{eqnarray}
For $h \ge 0$, define
\[{J_{xh}} \buildrel \Delta \over = \frac{1}{N}\sum\limits_{i = 1}^N {\sum\limits_{k = 1}^N {\int_0^h {{{\left( {{x_k}(t) - {x_i}(t)} \right)}^T}Q\left( {{x_k}(t) - {x_i}(t)} \right){\rm{d}}t} } } .\]
By (\ref{14}) and (\ref{15}), one can show that
\[
{J_{xh}} = \sum\limits_{i = 2}^N {\int_0^h {2\tilde x_i^T(t)Q{{\tilde x}_i}(t){\rm{d}}t} } .
\]
If $RA + {A^T}R - \gamma RB{B^T}R + 2Q \le 0$, then one has
\[
\hspace{-3em}{J_{xh}} \hspace{-3pt}=\hspace{-3pt} \sum\limits_{i = 2}^N {\int_0^h {\hspace{-4pt}2\tilde x_i^T(t)Q{{\tilde x}_i}(t){\rm{d}}t} }\hspace{-2pt} +\hspace{-3pt} \int_0^h \hspace{-4pt}{\dot V(t)} {\rm{d}}t\hspace{-2pt} -\hspace{-3pt} V(h) \hspace{-2pt}+ \hspace{-3pt}V(0)\vspace{-10pt}
\]
\begin{eqnarray}\label{16}
\hspace{1em} \le\hspace{-2pt}  - \hspace{-0pt}\gamma \sum\limits_{i = 2}^N {\int_0^{ + \infty }\hspace{-4pt} {\tilde x_i^T(t)RB{B^T}R{{\tilde x}_i}(t){\rm{d}}t} } \hspace{-2pt}+ \hspace{-3pt}V(0) \hspace{-3pt}-\hspace{-3pt} V(h).
\end{eqnarray}
Due to $\zeta (t) = \left[ {{{\bf{0}}_{(N - 1)d \times d}},{I_{(N - 1)d}}} \right]\left( {U_{\sigma (t)}^T \otimes {I_d}} \right)x(t)$ and ${\tilde U_{\sigma (t)}}\tilde U_{\sigma (t)}^T = {I_N} - {N^{ - 1}}{{\bf{1}}_N}{\bf{1}}_N^T$, one can show that
\[
\hspace{-16em}{\zeta ^T}(0)\left( {{I_{N - 1}} \otimes R} \right)\zeta (0) \vspace{-10pt}
\]
\begin{eqnarray}\label{17}
\hspace{4em}= {x^T}(0)\left( {\left( {{I_N} - \frac{1}{N}{{\bf{1}}_N}{\bf{1}}_N^T} \right) \otimes R} \right)x(0).
\end{eqnarray}
Due to ${w_{\sigma (0),ik}}(0) = 1$ and ${K_w} = RB{B^T}R \ge 0$, one has
\[
\hspace{-8em}\sum\limits_{i = 1}^N \sum\limits_{k \in {N_{\sigma (0),i}}} \frac{{{{\left( {{w_{\sigma (0),ik}}(0) + {l_{\sigma (0),ik}}} \right)}^2}}}{2} \vspace{-10pt}
\]
\begin{eqnarray}\label{18}
\hspace{3em}- \sum\limits_{i = 1}^N {\sum\limits_{k \in {N_{\sigma (h),i}}} {\frac{{{{\left( {{w_{\sigma (h),ik}}(h) + {l_{\sigma (h),ik}}} \right)}^2}}}{2}} }    \le 0.
\end{eqnarray}
Due to ${\lim _{t \to  + \infty }}\left( {{\gamma _{ik}} - {w_{\sigma (t),ik}}(t)} \right) = 0$, one can show that
\begin{eqnarray}\label{19}
\mathop {\lim }\limits_{h \to  + \infty } \sum\limits_{i = 1}^N {\sum\limits_{k = 1,k \ne i}^N {\left( {{\gamma _{ik}} - {w_{\sigma (h),ik}}(h)} \right)} }  = 0,
\end{eqnarray}
\begin{eqnarray}\label{20}
\sum\limits_{i = 1}^N {\sum\limits_{\scriptstyle k = 1 \hfill \atop
  \scriptstyle k \ne i \hfill}^N {\hspace{-2pt}\left( {{\gamma _{ik}} \hspace{-2pt}- \hspace{-2pt}{w_{\sigma (0),ik}}(0)} \right)} }  \hspace{-2pt}=\hspace{-2pt} \sum\limits_{i = 1}^N {\sum\limits_{\scriptstyle k = 1 \hfill \atop
  \scriptstyle k \ne i \hfill}^N {\int_0^{ + \infty } \hspace{-2pt}{{{\dot w}_{\sigma (t),ik}}(t)} {\rm{d}}t} } .
\end{eqnarray}
Since
\[
\sum\limits_{i = 1}^N {\sum\limits_{\scriptstyle k = 1 \hfill \atop
  \scriptstyle k \ne i \hfill}^N {\int_0^{ + \infty } \hspace{-4pt} {{{\dot w}_{\sigma (t),ik}}(t)} {\rm{d}}t} }   \hspace{-2pt}=  \hspace{-2pt}2N \hspace{-2pt}\int_0^{ + \infty } \hspace{-6pt}{{x^T}(t)\left( {{L_N} \hspace{-2pt} \otimes \hspace{-2pt} {K_w}} \right)x(t)} {\rm{d}}t,
\]
it can be derived by (\ref{20}) that
\[
\hspace{-10em}\frac{\gamma }{N}\sum\limits_{i = 1}^N {\sum\limits_{k = 1,k \ne i}^N {\left( {{\gamma _{ik}} - {w_{\sigma (0),ik}}(0)} \right)} } \vspace{-10pt}
\]
\begin{eqnarray}\label{21}
\hspace{5em} = 2\gamma \sum\limits_{i = 2}^N {\int_0^{ + \infty } {\tilde x_i^T(t)RB{B^T}R{{\tilde x}_i}(t){\rm{d}}t.} }
\end{eqnarray}
Let $h \to  + \infty $, then one can set from (\ref{16})-(\ref{19}) and (\ref{21}) that
\[
\hspace{-6.5em}{J^ * } = {x^T}(0)\left( {\left( {{I_N} - \frac{1}{N}{{\bf{1}}_N}{\bf{1}}_N^T} \right) \otimes R} \right)x(0) \vspace{-4pt}
\]
\[
\hspace{2.5em}{\rm{ + }}\gamma \int_0^{ + \infty }\hspace{-4pt} {{x^T}(t)\left(\hspace{-2pt} {\left( {{I_N} \hspace{-2pt}-\hspace{-2pt} \frac{1}{N}{{\bf{1}}_N}{\bf{1}}_N^T} \right)\hspace{-2pt} \otimes\hspace{-2pt} RB{B^T}R} \right)\hspace{-2pt}x(t)} {\rm{d}}t.
\]
Thus, the conclusion of Theorem \ref{theorem1} is obtained.
\end{proof}

If $(A,B)$ is stabilizable, then the Riccati equation $RA + {A^T}R - \gamma RB{B^T}R + 2Q = 0$ has a unique and positive definite solution $R$ for any given $\gamma  > 0$ as shown in \cite{c31}. In this case, the \emph{are} solver in the Matlab toolbox can be used to solve this Riccati equation. It should be pointed out that $\gamma$ represents the rightward translated quantity of the eigenvalues of ${L_{\sigma (t),0}}$, which can be previously given. Intuitionally speaking, because $ - RB{B^T}R$ is negative semidifinite, a large $\gamma$ can decrease the consensus control gain and the guaranteed-performance cost.\par

Furthermore, the consensus control gain and the guaranteed-performance cost can be regulated by introducing the gain factor $\varepsilon  > 0$ with $R \le \varepsilon I$, which means that $RB{B^T}R \le {\varepsilon ^2}B{B^T}$ if the maximum eigenvalue of $B{B^T}$ is not larger than $1$; that is, ${\lambda _{\max }}(B{B^T}) \le 1$. In this case, $\varepsilon $ can also be regarded as the maximum nonzero eigenvalue of $R$. In this case, both $\gamma$ and $R$ are variables, so it is difficult to determine the solution of $RA + {A^T}R - \gamma RB{B^T}R + 2Q \le 0$. Based on LMI techniques, by Schur complement lemma in \cite{c32}, the following corollary presents an adaptive guaranteed-performance consensualization criterion with a given gain factor, which can be solved by the \emph{feasp} solver in the LMI toolbox.\par

\begin{col}\label{corollary1}
For any given gain factor $\varepsilon  > 0$, multiagent system (\ref{1}) is adaptively guaranteed-performance consensualizable by protocol (\ref{2}) if ${\lambda _{\max }}(B{B^T}) \le 1$ and there exist $\gamma  > 0$ and ${\tilde R^T} = \tilde R \ge {\varepsilon ^{ - 1}}I$ such that
\[\tilde \Xi  = \left[ {\begin{array}{*{20}{c}}
   {A\tilde R + \tilde R{A^T} - \gamma B{B^T}} & {2\tilde RQ}  \\
   $*$ & { - 2Q}  \\
\end{array}} \right] < 0.\]
In this case, ${K_u} = {B^T}{\tilde R^{ - 1}}$, ${K_w} = {\tilde R^{ - 1}}B{B^T}{\tilde R^{ - 1}}$ and the guaranteed-performance cost satisfies that
\[
{J^ * } = \sum\limits_{i = 2}^N {\left( {\varepsilon {{\left\| {{{\tilde x}_i}(0)} \right\|}^2} + \gamma {\varepsilon ^2}\int_0^{ + \infty } {{{\left\| {{B^T}{{\tilde x}_i}(t)} \right\|}^2}} {\rm{d}}t} \right)} .
\]
\end{col}
\vspace{6pt}
\begin{rmk}\label{remark2}
In \cite{c24}, guaranteed-performance consensus for multiagent systems with time-varying neighbor sets and time-invariant interaction strengths was investigated, where the minimum and maximum nonzero eigenvalues of the Laplacian matrices of all interaction topologies in the switching set are required to design gain matrices of consensus protocols. It should be pointed out that global structure information of interaction topologies of the whole system is required to determine the precise values of the minimum and maximum nonzero eigenvalues. Moreover, their methods cannot deal with time-varying interaction strength cases. By the translation-adaptive strategy, the impacts of both the minimum and maximum nonzero eigenvalues and time-varying interaction strengths are eliminated in Theorem 1 and Corollary 1, and a completely distributed guaranteed-performance consensus control is realized in the sense that consensualization criteria are independent of the Laplacian matrices of interaction topologies in the switching set and their eigenvalues.
\end{rmk}

\begin{rmk}\label{remark3}
The scaling strategy was applied to realize adaptive consensus control in \cite{c28,c29,c30}, where the impacts of switching topologies were not investigated. The scaling factor in the Lyapunov function is inversely proportional to the minimum nonzero eigenvalue of the Laplacian matrix, so this factor is difficult to be determined and may be very large since the minimum nonzero eigenvalue may be very small. Thus, the consensus regulation performance cannot be guaranteed since the scaling factor in the Lyapunov function cannot be eliminated when dealing with guaranteed-performance constraints. By translating all nonzero eigenvalues of the Laplacian matrices in the switching set rightward instead of the scaling factor, the guaranteed-performance constraints can be dealt with and the guaranteed-performance cost consists of two terms: the initial state term $J_{x(0)}^*$ and the state integral term $J_{x(t)}^*$. The guaranteed-performance cost only contain the initial state term in \cite{c23,c24,c25,c26,c27}, where the adaptive consensus strategy was not applied. Actually, the state integral term is introduced since the interaction strengths are adaptively adjusted.
\end{rmk}

\section{Extensions to Lipschitz nonlinear cases }\label{section4}
This section extends adaptive guaranteed-performance consensualization criteria for linear multiagent systems shown in the above section to multiagent systems with each agent containing the Lipschitz nonlinearity.\par

The dynamics of each agent is modeled as
\begin{eqnarray}\label{22}
{\dot x_i}(t) = A{x_i}(t) + f({x_i}(t)) + B{u_i}(t){\rm{ }}(i \in \{ 1,2, \cdots ,N\} ),
\end{eqnarray}
where the nonlinear function $f:{\mathbb{R}^d} \times \left[ {0, + \infty } \right) \to {\mathbb{R}^d}$ is continuous and differentiable and satisfies the Lipschitz condition $\left\| {f({x_i}(t)) - f({x_k}(t))} \right\| \le \mu \left\| {{x_i}(t) - {x_k}(t)} \right\|$ with the Lipschitz constant $\mu  > 0$, and all the other notations are identical with the ones in (\ref{1}). Let $F(x(t)) = {\left[ {{f^T}\left( {{x_1}(t)} \right),{f^T}\left( {{x_2}(t)} \right), \cdots ,{f^T}\left( {{x_N}(t)} \right)} \right]^T}{\rm{ }}{\rm{.}}$ By the similar analysis in the above section, the dynamics of multiagent system (\ref{22}) with protocol (\ref{2}) can be transformed into
\begin{eqnarray}\label{23}
{\dot {\tilde x}_1}(t) = A{\tilde x_1}(t) + \left( {\frac{1}{{\sqrt N }}{\bf{1}}_N^T \otimes {I_d}} \right)F(x(t)),
\end{eqnarray}
\[
\hspace{-1em}\dot \zeta (t) = \left( {{I_{N - 1}} \otimes A - \tilde U_{\sigma (t)}^T{L_{\sigma (t),w}}{{\tilde U}_{\sigma (t)}} \otimes B{K_u}} \right)\zeta (t)\vspace{-12pt}
\]
\begin{eqnarray}\label{24}
\hspace{-5em}+ \left( {\tilde U_{\sigma (t)}^T \otimes {I_d}} \right)F(x(t)),
\end{eqnarray}
where subsystem (\ref{24}) describes the disagreement dynamics of multiagent system (\ref{22}). \par

By the Lipschitz condition and the structure feature of the transformation matrix ${\tilde U_{\sigma (t)}} \otimes {I_d}$, the following theorem linearizes the impacts of the nonlinear term $\left( {\tilde U_{\sigma (t)}^T \otimes {I_d}} \right)F(x(t))$ and gives an adaptive guaranteed-performance consensualization criterion in the completely distributed manner; that is, it is not associated with Laplacian matrices of interaction topologies in the switching set and their eigenvalues. \par

\begin{thm}\label{theorem2}
For any given $\gamma  > 0$, multiagent system (\ref{22}) is adaptively guaranteed-performance consensualizable by protocol (\ref{2}) if there exists a matrix ${P^T} = P > 0$ such that
\[
PA + {A^T}P - P(\gamma B{B^T} - I)P + 2Q + {\mu ^2}I \le 0.
\]
In this case, ${K_u} = {B^T}P$, ${K_w} = PB{B^T}P$ and the guaranteed-performance cost satisfies that ${J^ * } = J_{x(0)}^* + J_{x(t)}^*,$ where
\[
J_{x(0)}^* = {x^T}(0)\left( {\left( {{I_N} - \frac{1}{N}{{\bf{1}}_N}{\bf{1}}_N^T} \right) \otimes P} \right)x(0),
\]
\[
J_{x(t)}^*{\rm{ = }}\gamma \int_0^{ + \infty }  \hspace{-2pt}{{x^T}(t)\left( {\left( {{I_N}  \hspace{-2pt}- \hspace{-2pt} \frac{1}{N}{{\bf{1}}_N}{\bf{1}}_N^T} \right) \hspace{-2pt} \otimes \hspace{-2pt}PB{B^T}P} \right)x(t)} {\rm{d}}t.
\] \hspace{4pt}
\end{thm}

\begin{proof}
Due to ${\tilde U_{\sigma (t)}}\tilde U_{\sigma (t)}^T = {L_N}$, it can be shown that
\[
\hspace{-9em}{F^T}(x(t))\left( {{{\tilde U}_{\sigma (t)}}\tilde U_{\sigma (t)}^T \otimes {I_d}} \right){\rm{ }}F(x(t)) \vspace{-6pt}
\]
\[
\hspace{6em} = \frac{1}{{2N}}\sum\limits_{i = 1}^N {\sum\limits_{k = 1}^N {{{\left\| {f({x_i}(t)) - f({x_k}(t))} \right\|}^2}} }.
\]
By the Lipschitz condition, one can see that
\[
\frac{1}{{2N}}\hspace{-2pt}\sum\limits_{i = 1}^N \hspace{-3pt}{\sum\limits_{k = 1}^N {{{\left\| \hspace{-1pt}{f({x_i}(t)\hspace{-1pt}) \hspace{-3pt}-\hspace{-3pt} f({x_k}(t)\hspace{-1pt})}\hspace{-1pt} \right\|}^2}} }  \hspace{-3pt}\le\hspace{-2pt} \frac{{{\mu ^2}}}{{2N}}\hspace{-2pt}\sum\limits_{i = 1}^N\hspace{-3pt} {\sum\limits_{k = 1}^N {{{\left\| \hspace{-1pt}{{x_i}(t) \hspace{-3pt}- \hspace{-3pt}{x_k}(t)}\hspace{-1pt} \right\|}^2}} }\vspace{0pt}
\]
\[
\hspace{5em} = \hspace{-2pt}{\mu ^2}{x^T}(t)\hspace{-2pt}\left( {{{\tilde U}_{\sigma (t)}}\tilde U_{\sigma (t)}^T\hspace{-2pt} \otimes\hspace{-2pt} {I_d}} \right)\hspace{-2pt}{\rm{ }}x(t) \vspace{-10pt}
\]
\begin{eqnarray}\label{25}
\hspace{2.5em} = {\mu ^2}\sum\limits_{i = 2}^N {\tilde x_i^T(t){{\tilde x}_i}(t)} .
\end{eqnarray}
It can be derived that
\[
\hspace{-4em}2{\zeta ^T}(t)\left( {\tilde U_{\sigma (t)}^T \otimes P} \right)F\left( {x(t)} \right)  \le \sum\limits_{i = 2}^N {\tilde x_i^T(t)PP{{\tilde x}_i}(t)} \vspace{-10pt}
\]
\begin{eqnarray}\label{26}
\hspace{6em}+ {F^T}\left( {x(t)} \right)\left( {{{\tilde U}_{\sigma (t)}}\tilde U_{\sigma (t)}^T \otimes {I_d}} \right){\rm{ }}F\left( {x(t)} \right).
\end{eqnarray}
By constructing a similar Lyapunov function in (\ref{10}) with $R$ replacing by $P$, from (\ref{25}) and (\ref{26}), the conclusion of Theorem 2 can be obtained.
\end{proof}

For any given gain factor $\varepsilon  > 0$ with $P \le \varepsilon I$, the adaptive guaranteed-cost consensualization can be realized by choosing proper $\gamma $ and $P$. According to Theorem 2 and Schur complement lemma, the following corollary presents an approach to determine gain matrices of consensus protocols with a given gain factor in terms of LMIs.\par

\begin{col}\label{corollary2}
For any given $\varepsilon  > 0$, multiagent system (\ref{22}) is adaptively guaranteed-cost consensualizable by protocol (\ref{2}) if ${\lambda _{\max }}(B{B^T}) \le 1$ and there exist $\gamma  > 0$ and ${\tilde P^T} = \tilde P \ge {\varepsilon ^{ - 1}}I$ such that
\[
\hat \Xi  = \left[ {\begin{array}{*{20}{c}}
   {A\tilde P + \tilde P{A^T} - \gamma B{B^T} + I} & {2\tilde PQ} & {\mu \tilde P}  \\
   $*$ & { - 2Q} & 0  \\
   $*$ & $*$ & { - I}  \\
\end{array}} \right] < 0.
\]
In this case, ${K_u} = {B^T}{\tilde P^{ - 1}}$, ${K_w} = {\tilde P^{ - 1}}B{B^T}{\tilde P^{ - 1}}$ and the guaranteed-performance cost satisfies that
\[{
J^ * } = \sum\limits_{i = 2}^N {\left( {\varepsilon {{\left\| {{{\tilde x}_i}(0)} \right\|}^2} + \gamma {\varepsilon ^2}\int_0^{ + \infty } {{{\left\| {{B^T}{{\tilde x}_i}(t)} \right\|}^2}} {\rm{d}}t} \right)} .
\]
\vspace{0pt}
\end{col}

We adopt two critical approaches to deduce our main conclusions: the variable changing approach and the Riccati inequality approach. It should be pointed out that the variable changing approach is an equivalent transformation, so it does not introduce any conservatism. However, since there exists some scalability of the Lyapunov function, the Riccati inequality approach may bring in some conservatism. Actually, the Riccati inequality approach is extensively applied in optimization control and often has less conservatism as shown in \cite{c31}. Moreover, two key difficulties exist in obtaining the main result of the current paper shown in Theorems \ref{theorem1} and \ref{theorem2}. The first one is to design a proper Lyapunov function, which can be used to rightward translate the nonzero eigenvalues of the Laplacian matrix. The second one is to determine the relationship between the Laplacian matrix and the linear quadratic index, as given in (\ref{14}) and (\ref{15}).

\begin{rmk}\label{remark4}
Many multiagent systems contain Lipschitz nonlinear dynamics. For example, sinusoidal terms are globally Lipschitz, which are usually encountered in cooperative control for multiple robotics and cooperative guidance for multiple unmanned vehicles as shown in \cite{c33} and \cite{c34}. The key difficulties contain two aspects: how to decompose the disagreement components from the whole $F(x(t))$ and how to eliminate the impacts of the time-varying transformation matrix ${\tilde U_{\sigma (t)}} \otimes {I_d}$. Because ${\tilde U_{\sigma (t)}}\tilde U_{\sigma (t)}^T$ is the Laplacian matrix of a complete graph with the weights of all the edges ${1 \mathord{\left/ {\vphantom {1 N}} \right. \kern-\nulldelimiterspace} N}$, these two key challenges can be dealt with by using this special structure characteristic. Moreover, since the consensus regulation performance and adaptively adjusting interaction weights are considered, the approaches to deal with Lipschitz nonlinear dynamics in \cite{c16,c17,c18} are no longer valid.
\end{rmk}

\section{Numerical simulations}\label{section5}
In this section, two simulation examples are provided to demonstrate the theoretical results obtained in the previous sections. \par

\vspace{1em}
\centerline {***** Put Fig. 1 about here *****}
\vspace{1em}

\begin{exm}[Linear cases]
Consider a four-order linear multiagent system composed of six agents with switching interaction topologies ${G_1}$, ${G_2}$, ${G_3}$ and ${G_4}$ given in Fig. 1. The dynamics of each agent is shown in (\ref{1}) with
\[A \! = \! \left[ {\begin{array}{*{20}{c}}
   {{\rm{ - 3}}{\rm{.375}}} & {{\rm{ - 4}}{\rm{.500}}} & {{\rm{ - 4}}{\rm{.125}}} & {{\rm{ - 3}}{\rm{.250}}}  \\
   {{\rm{1}}{\rm{.625}}} & {{\rm{ - 1}}{\rm{.500}}} & {{\rm{ - 1}}{\rm{.125}}} & {{\rm{1}}{\rm{.250}}}  \\
   {{\rm{ - 0}}{\rm{.875}}} & {{\rm{0}}{\rm{.500}}} & {{\rm{ - 1}}{\rm{.625}}} & {{\rm{1}}{\rm{.750}}}  \\
   {{\rm{1}}{\rm{.750}}} & {{\rm{3}}{\rm{.500}}} & {{\rm{2}}{\rm{.750}}} & {{\rm{ - 0}}{\rm{.500}}}  \\
\end{array}} \right] \! \! ,B \! = \! \left[ {\begin{array}{*{20}{c}}
   0  \\
   {1.5}  \\
   0  \\
   0  \\
\end{array}} \right]\! \! .\]
The initial state of each agent is
\[\begin{array}{l}
 {x_1}\left( 0 \right) = {\left[ { - 1,0, - 4,5} \right]^T},~~{x_2}\left( 0 \right) = {\left[ {5, - 4, - 8, - 2} \right]^T}, \\
 {x_3}\left( 0 \right) = {\left[ { - 7,3, - 4,7} \right]^T},~~{x_4}\left( 0 \right) = {\left[ { - 2,7, - 1, - 5} \right]^T}, \\
 {x_5}\left( 0 \right) = {\left[ {4,6, - 1, - 3} \right]^T},~~{x_6}\left( 0 \right) = {\left[ {8,1,5, - 4} \right]^T}. \\
 \end{array}\]

The parameters are chosen as $\gamma  = 5$ and
\[
Q = \left[ {\begin{array}{*{20}{c}}
   {0.10} & {0.02} & {0.01} & 0  \\
   {0.02} & {0.10} & {0.01} & {0.02}  \\
   {0.01} & {0.01} & {0.10} & {0.03}  \\
   0 & {0.02} & {0.03} & {0.10}  \\
\end{array}} \right],
\]
then one can obtain gain matrices according to Theorem \ref{theorem1} as follows
\[
{K_u} = \left[ {\begin{array}{*{20}{c}}
   {{\rm{0}}{\rm{.2}}653} & {{\rm{1}}{\rm{.0549}}} & {0.7878} & {0.6790}  \\
\end{array}} \right],
\]
\[
{K_w} = \left[ {\begin{array}{*{20}{c}}
   {0.0704} & {0.2799} & {0.2090} & {0.1801}  \\
   {0.2799} & {1.1128} & {0.8311} & {0.7163}  \\
   {0.2090} & {0.8311} & {0.6207} & {0.5349}  \\
   {0.1801} & {0.7163} & {0.5349} & {0.4610}  \\
\end{array}} \right].
\]
In this case, the guaranteed-performance cost is ${J^ * } = {\rm{267}}{\rm{.9357}}$. \par

As illustrated in Fig. 2, let ${G_{\sigma (t)}}$ randomly switch among ${G_1}$, ${G_2}$, ${G_3}$ and ${G_4}$ with switching interval $0.5$s. Figs. 3 and 4 depict the guaranteed-performance function ${J_{x\left( t \right)}}$ and the trajectories of the states of all agents ${x_i}(t)$ $\left( {i = 1,2, \cdots ,N} \right)$, respectively. One can see that multiagent system (1) achieves adaptive consensus and the guaranteed-performance function ${J_{x\left( t \right)}}$ converges to a finite value with ${J_{x\left( t \right)}} < {J^ * }$. The simulation results illustrate that multiagent system (1) can be adaptively guaranteed-performance consensualizable by protocol (2) with the above gain matrices ${K_u}$ and ${K_w}$ obtained by Theorem \ref{theorem1} without using the global information of the interaction topology. However, the distributed consensus control approach in \cite{c35} required the precise value of the minimum nonzero eigenvalue of the interaction topology; that is, the completely distributed control cannot be realized. Moreover, the main conclusion of Theorem 1 is completely distributed, so it should be pointed out that the computational complexity does not increase as the number of agents increases.\par

\vspace{1em}
\centerline {***** Put Fig. 2 about here *****}
\vspace{1em}

\vspace{1em}
\centerline {***** Put Fig. 3 about here *****}
\vspace{1em}

\vspace{1em}
\centerline {***** Put Fig. 4 about here *****}
\vspace{1em}
\end{exm}

\begin{exm}[Lipschitz nonlinear cases]
Consider a four-order Lipschitz nonlinear multiagent system with six agents and the dynamics of each agent is described by (\ref{22}) with
\[\begin{array}{l}
 A = \left[ {\begin{array}{*{20}{c}}
   {{\rm{ - 3}}{\rm{.125}}} & {{\rm{ - 5}}{\rm{.250}}} & {{\rm{ - 4}}{\rm{.625}}} & {{\rm{ - 4}}{\rm{.250}}}  \\
   {{\rm{0}}{\rm{.875}}} & {{\rm{ - 2}}{\rm{.250}}} & {{\rm{ - 2}}{\rm{.625}}} & {{\rm{ - 0}}{\rm{.750}}}  \\
   {{\rm{ - 0}}{\rm{.625}}} & {{\rm{1}}{\rm{.750}}} & {{\rm{ - 0}}{\rm{.125}}} & {{\rm{2}}{\rm{.750}}}  \\
   {{\rm{2}}{\rm{.250}}} & {{\rm{3}}{\rm{.000}}} & {{\rm{2}}{\rm{.750}}} & {{\rm{0}}{\rm{.500}}}  \\
\end{array}} \right], \\
 B = \left[ {\begin{array}{*{20}{c}}
   0  \\
   1  \\
   0  \\
   0  \\
\end{array}} \right],f\left( {{x_i}} \right) = \left[ {\begin{array}{*{20}{c}}
   0  \\
   0  \\
   0  \\
   { - \mu \sin \left( {{x_{i3}}} \right)}  \\
\end{array}} \right], \\
 \end{array}\]
where ${x_i} = {\left[ {{x_{i1}},{\rm{ }}{x_{i2}},{\rm{ }}{x_{i3}},{\rm{ }}{x_{i4}}} \right]^T}$ $\left( {i = 1,2, \cdots ,6} \right)$ and $\mu  = 0.0333$. The initial states of all agents are given as
\[\begin{array}{l}
 {x_1}\left( 0 \right) = {\left[ { - 1, - 2, - 3,5} \right]^T},\ \ \; {x_2}\left( 0 \right) = {\left[ { - 0.5,2, - 4,1.6} \right]^T}, \\
 {x_3}\left( 0 \right) = {\left[ {6, - 3,2,3} \right]^T},\qquad \ \, {x_4}\left( 0 \right) = {\left[ { - 2.5,2,3, - 5} \right]^T}, \\
 {x_5}\left( 0 \right) = {\left[ {1.7, - 9,1.5, - 3} \right]^T},{\kern 1pt} {x_6}\left( 0 \right) = {\left[ { - 1,4, - 2, - 6} \right]^T}. \\
 \end{array}\]

Let $\varepsilon  = 5$ and
\[
Q = \left[ {\begin{array}{*{20}{c}}
   {0.20} & {0.02} & {0.01} & 0  \\
   {0.02} & {0.1} & {0.03} & {0.02}  \\
   {0.01} & {0.03} & {0.2} & {0.03}  \\
   0 & {0.02} & {0.03} & {0.10}  \\
\end{array}} \right],
\]
then one can obtain from Corollary 2 that
\[\gamma  = 21.1207,\]
\[
{K_u} = \left[ {\begin{array}{*{20}{c}}
   {{\rm{0}}{\rm{.0989}}} & {{\rm{0}}{\rm{.6246}}} & {0.5940} & {0.4970}  \\
\end{array}} \right],
\]
\[
{K_w} = \left[ {\begin{array}{*{20}{c}}
   {0.0098} & {0.0618} & {0.0587} & {0.0491}  \\
   {0.0618} & {0.3902} & {0.3710} & {0.3104}  \\
   {0.0587} & {0.3710} & {0.3529} & {0.2952}  \\
   {0.0491} & {0.3104} & {0.2952} & {0.2470}  \\
\end{array}} \right],
\]
and the guaranteed-performance cost is that ${J^ * } = {\rm{4}}{\rm{.1478}} \times {\rm{1}}{{\rm{0}}^4}$. \par

Fig. 5 shows the switching signal $\sigma (t)$ and the switching set is also given in Fig. 1. Figs. 6 and 7 show the curves of the guaranteed-performance function and the state trajectories of this multiagent system, respectively. It can be found that the given Lipschitz nonlinear multiagent system (\ref{22}) can be adaptively guaranteed-cost consensualizable by protocol (\ref{2}) with ${J_{x\left( t \right)}} < {J^*}$.

\vspace{1em}
\centerline {***** Put Fig. 5 about here *****}
\vspace{1em}

\vspace{1em}
\centerline {***** Put Fig. 6 about here *****}
\vspace{1em}

\vspace{1em}
\centerline {***** Put Fig. 7 about here *****}
\vspace{1em}

Furthermore, for the case that $\varepsilon  = 10$ and the other parameters are identical, according to Corollary \ref{corollary2}, one can acquire that
\[
{K_u} = \left[ {\begin{array}{*{20}{c}}
   {{\rm{0}}{\rm{.1043}}} & {{\rm{0}}{\rm{.6793}}} & {0.6524} & {0.5448}  \\
\end{array}} \right],
\]
\[
{K_w} = \left[ {\begin{array}{*{20}{c}}
   {0.0109} & {0.0708} & {0.0680} & {0.0568}  \\
   {0.0708} & {0.4615} & {0.4432} & {0.3701}  \\
   {0.0680} & {0.4432} & {0.4257} & {0.3554}  \\
   {0.0568} & {0.3701} & {0.3554} & {0.2968}  \\
\end{array}} \right],
\]
\[
{J^ * } = {\rm{1}}{\rm{5.953}} \times {\rm{1}}{{\rm{0}}^4}.
\]
It can be seen that the gain matrices ${K_u}$, ${K_w}$ and guaranteed-performance cost ${J^ * }$ for $\varepsilon  = 10$ are larger than the ones for $\varepsilon  = 5$, which means that the values of ${K_u}$, ${K_w}$ and ${J^ * }$ become larger as $\varepsilon$ gets larger. Thus, one can obtain different ${K_u}$, ${K_w}$ and ${J^ * }$ to satisfy different requirements by regulating $\varepsilon$ in practical applications.\par
\end{exm}

\section{Conclusions}\label{section6}
A completely distributed guaranteed-performance consensus scheme was proposed to make consensus control gains independent of Laplacian matrices of switching topologies and their eigenvalues. An adaptive guaranteed-performance consensus design criterion for high-order linear multiagent systems with switching topologies was given based on the Riccati inequality, where the impacts of nonzero eigenvalues of Laplacian matrices of switching topologies with adaptively adjusting weights were eliminated by rightward translating them instead of scaling them. Furthermore, by adding constraints on the input matrix of each agent, it was shown that the consensus control gains and the guaranteed-performance cost can be regulated via choosing the different translation factor. Moreover, adaptive guaranteed-performance consensus conclusions for high-order linear multiagent systems were extended to high-order nonlinear ones by the Lipschitz condition and the structure characteristic of the transformation matrix. The further work is to investigate the influences of directed topologies, given cost budgets, and time-varying delays on adaptive guaranteed-performance consensus of multiagent systems with jointly connected switching topologies.

\newpage

\begin{figure}[!htb]
\begin{center}
\scalebox{0.8}[0.8]{\includegraphics{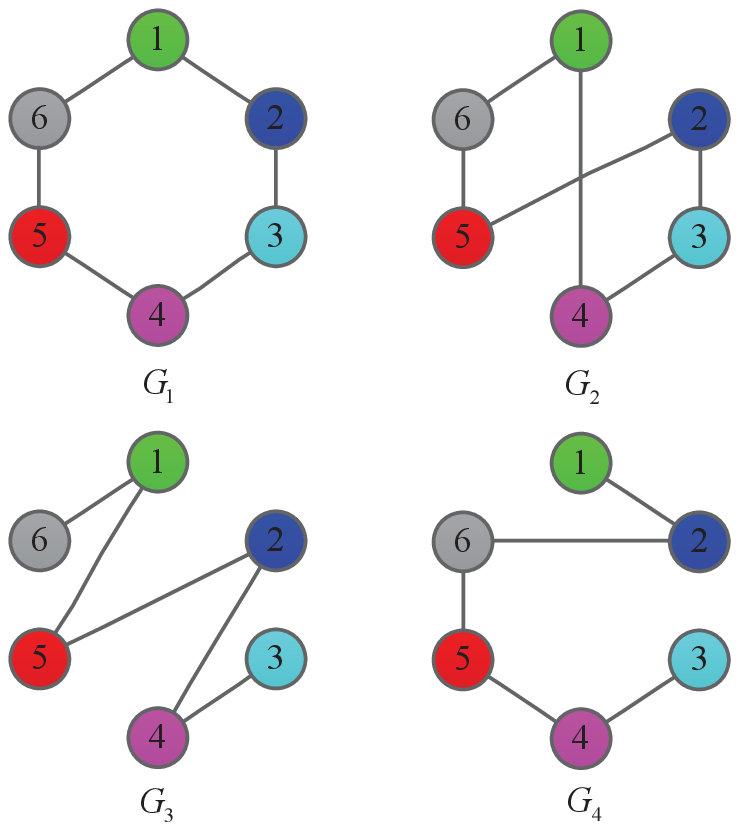}} \vspace{0pt}
\caption{Switching interaction topologies set.}\vspace{-10pt}
\end{center}\vspace{3em}
\end{figure}

\vspace{10pt}
\begin{figure}[!htb]
\begin{center}
\scalebox{0.5}[0.5]{\includegraphics{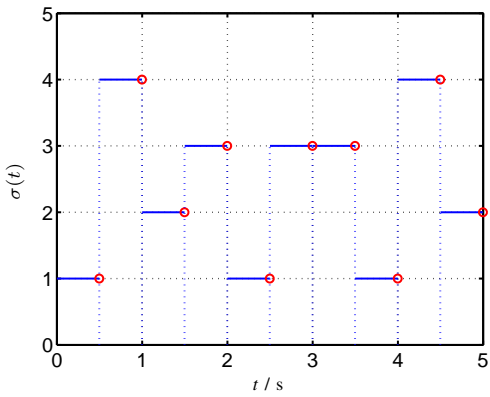}}
\put (-200, 72) {\rotatebox{90} {{\scriptsize $\sigma (t)$}}}
\put (-110, 0) {{ \scriptsize {\it t}~/~s}}
\caption{Switching signal $\sigma (t)$ for linear cases.}
\end{center}\vspace{3em}
\end{figure}

\vspace{10pt}
\begin{figure}[!htb]
\begin{center}
\scalebox{0.5}[0.5]{\includegraphics{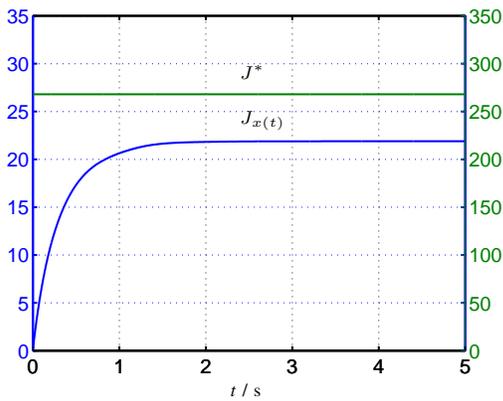}}
\put (-105, 120) {{\scriptsize ${J^ * }$}}
\put (-105, 103) {{\scriptsize ${J_{x\left( t \right)}}$}}
\put (-112, 0) {{ \scriptsize {\it t}~/~s}}
\caption{Guaranteed-performance function for linear cases.}
\end{center}\vspace{-1.5em}
\end{figure}

\vspace{10pt}
\begin{figure}[!htb]
\begin{center}
\scalebox{0.5}[0.5]{\includegraphics{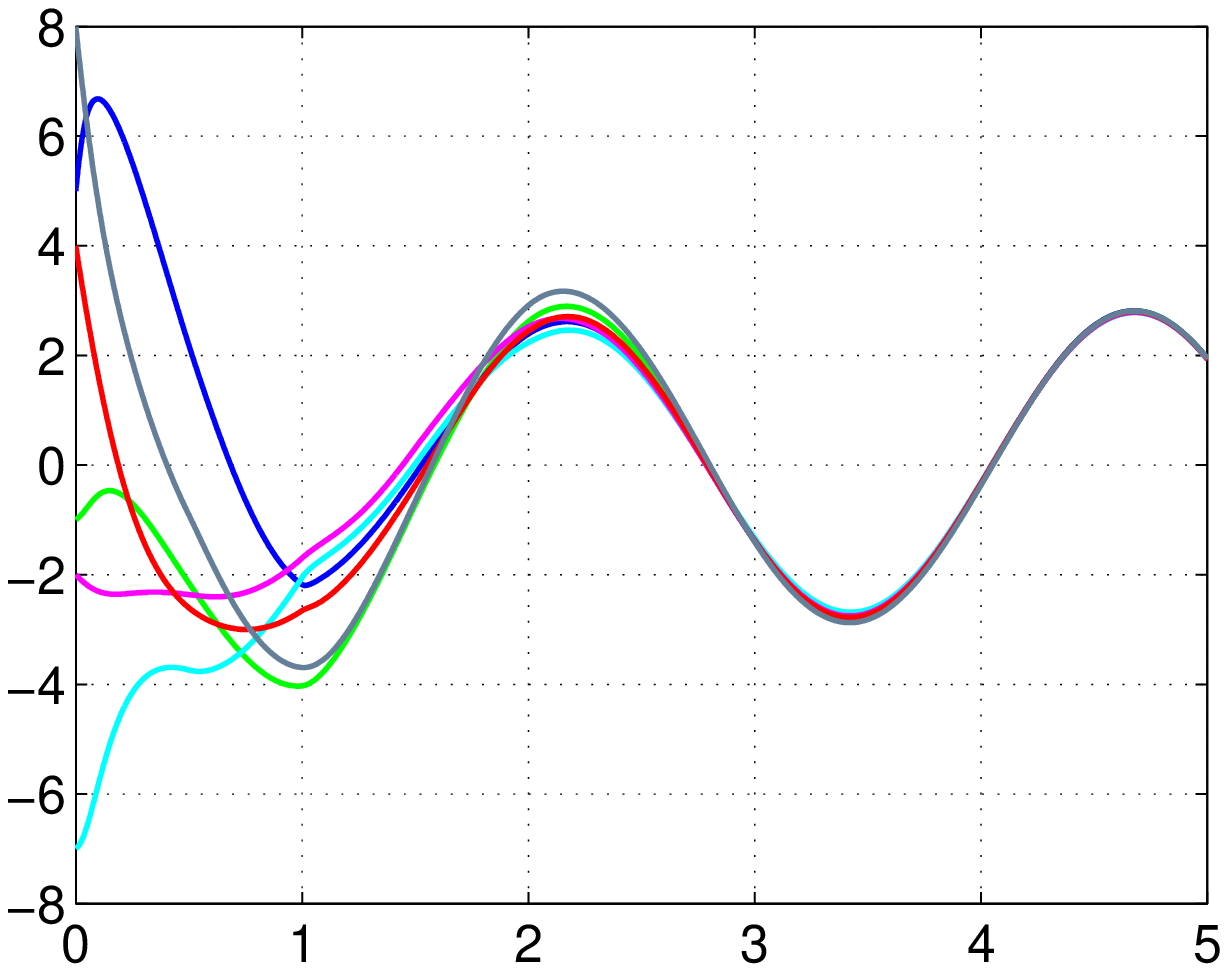}}
\put (-203, 42) {\rotatebox{90} {{\scriptsize $x_{i{\rm 1}}(t)~(i = 1,2, \cdots ,6)$}}}
\put (-110, 0) {{ \scriptsize {\it t}~/~s}}
\end{center}\vspace{-0.7em}
\end{figure}
\begin{figure}[!htb]
\begin{center}
\scalebox{0.5}[0.5]{\includegraphics{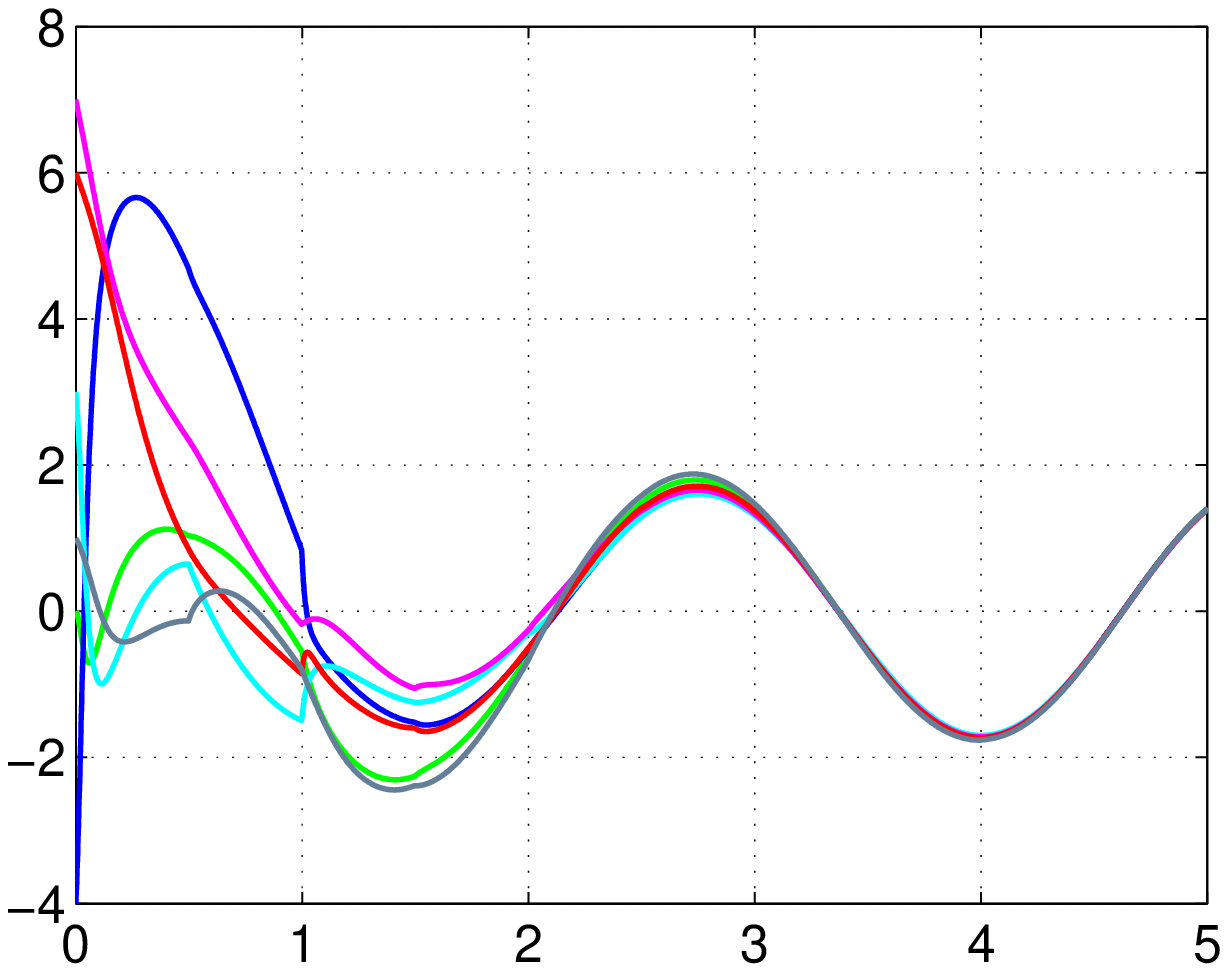}}
\put (-203, 42) {\rotatebox{90} {{\scriptsize $x_{i{\rm 2}}(t)~(i = 1,2, \cdots ,6)$}}}
\put (-110, 0) {{ \scriptsize {\it t}~/~s}}
\end{center}\vspace{-0.7em}
\end{figure}
\begin{figure}[!htb]
\begin{center}
\scalebox{0.5}[0.5]{\includegraphics{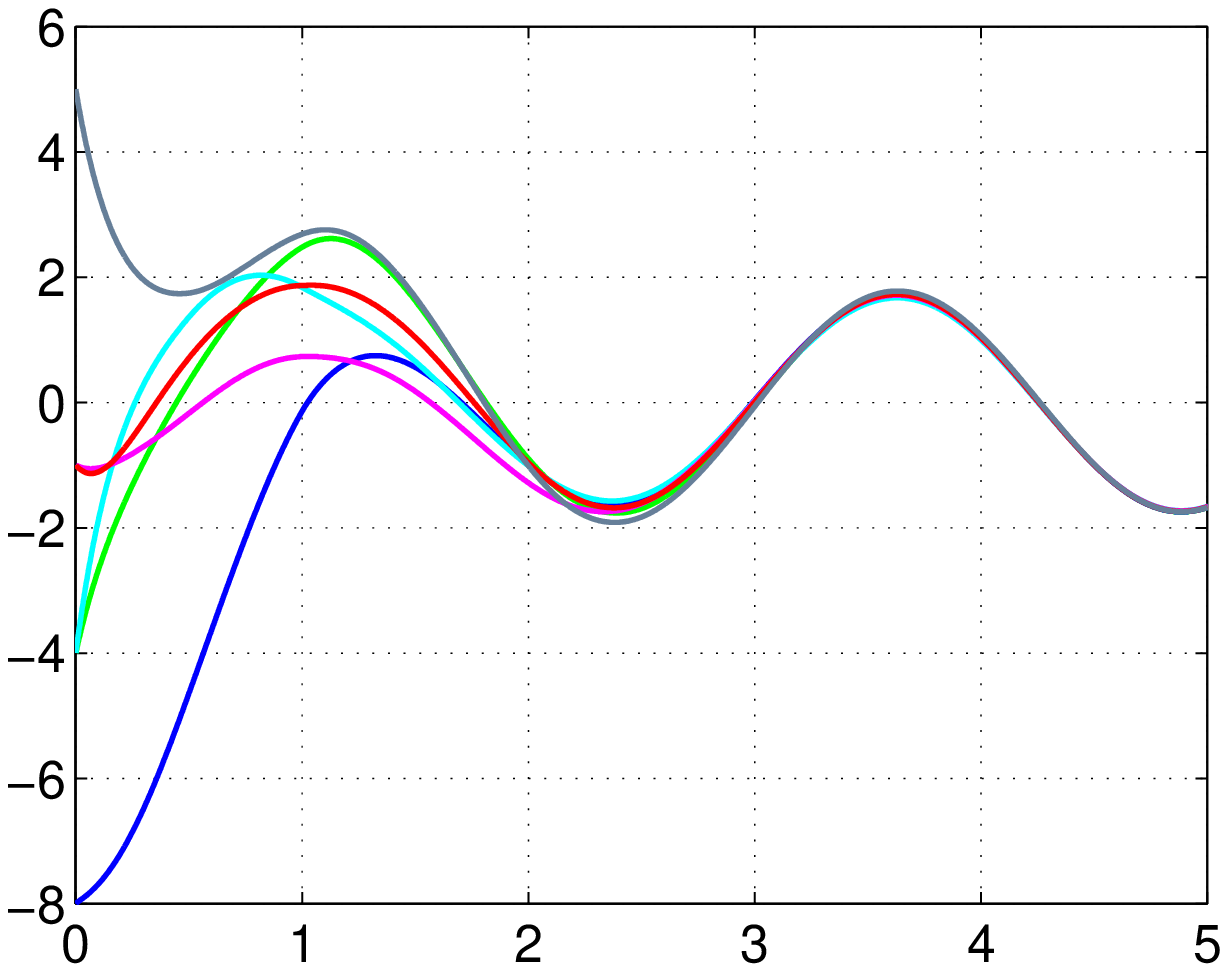}}
\put (-203, 42) {\rotatebox{90} {{\scriptsize $x_{i{\rm 3}}(t)~(i = 1,2, \cdots ,6)$}}}
\put (-110, 0) {{ \scriptsize {\it t}~/~s}}
\end{center}\vspace{-0.7em}
\end{figure}
\begin{figure}[!htb]
\begin{center}
\scalebox{0.5}[0.5]{\includegraphics{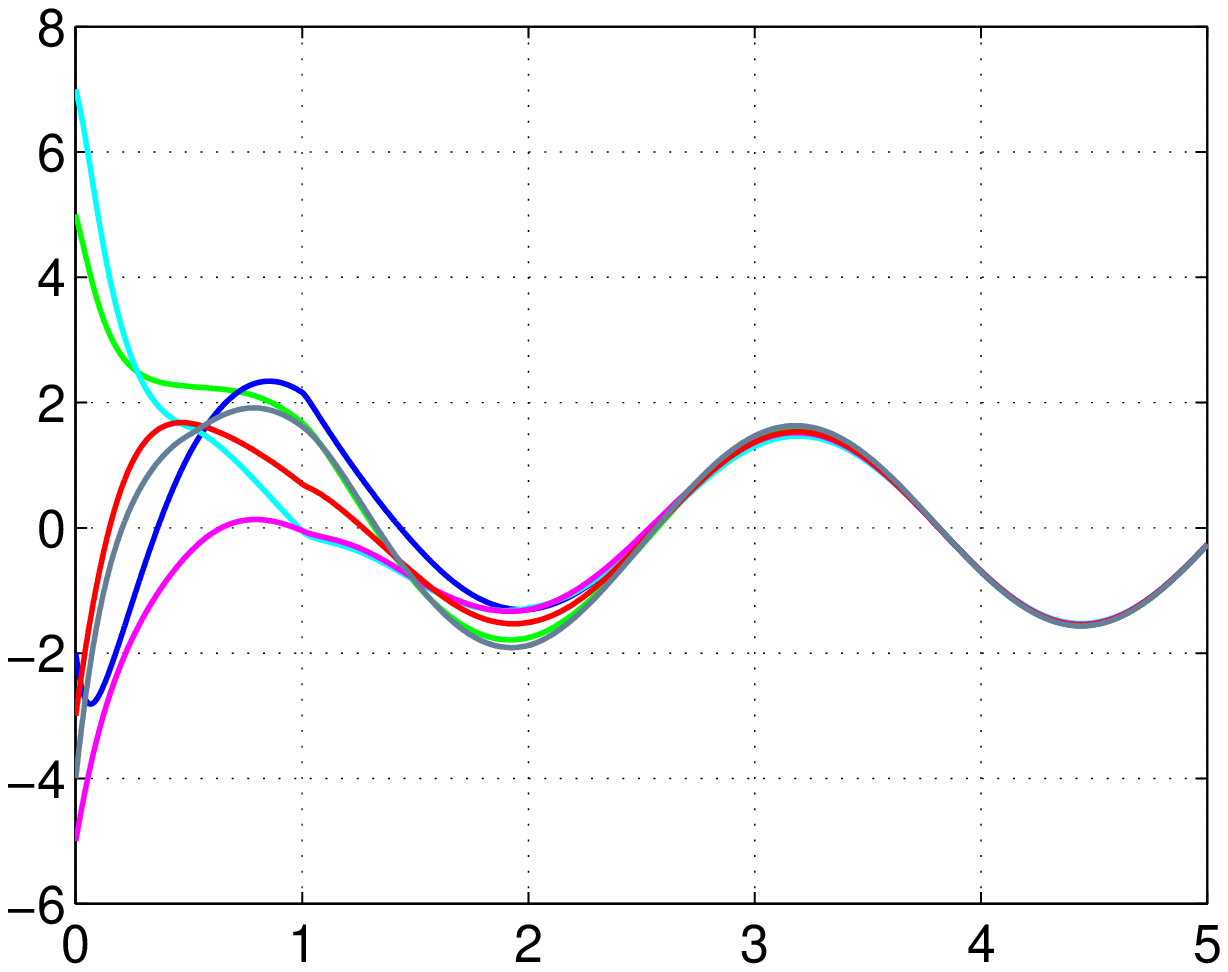}}
\put (-203, 42) {\rotatebox{90} {{\scriptsize $x_{i{\rm 4}}(t)~(i = 1,2, \cdots ,6)$}}}
\put (-110, 0) {{ \scriptsize {\it t}~/~s}}
\caption{State trajectories for linear cases.}
\end{center}
\end{figure}

\begin{figure}[!htb]
\begin{center}
\scalebox{0.5}[0.5]{\includegraphics{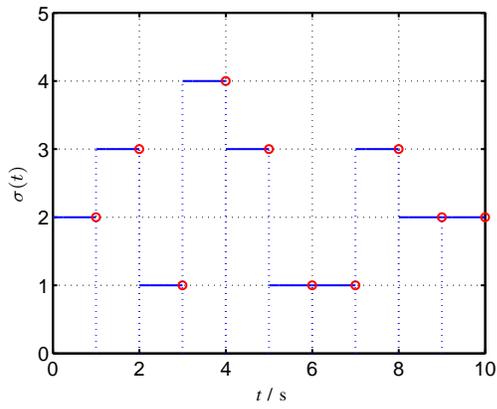}}
\put (-200, 75) {\rotatebox{90} {{\scriptsize $\sigma (t)$}}}
\put (-110, 0) {{ \scriptsize {\it t}~/~s}}
\caption{Switching signal $\sigma (t)$ for Lipschitz nonlinear cases.}
\end{center}\vspace{25em}
\end{figure}

\begin{figure}[!htb]
\begin{center}
\scalebox{0.5}[0.5]{\includegraphics{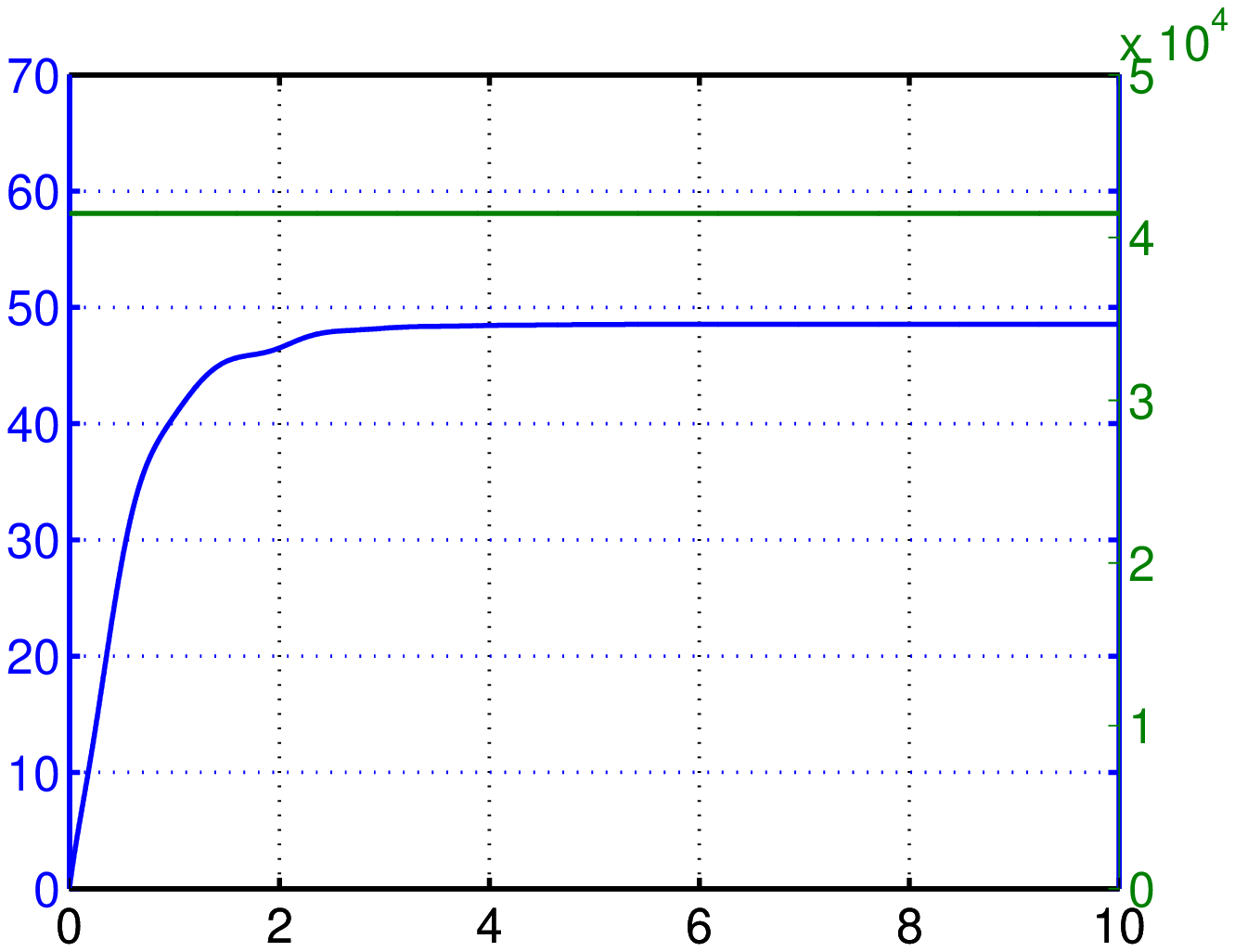}}
\put (-105, 128) {{\scriptsize ${J^ * }$}}
\put (-105, 111) {{\scriptsize ${J_{x\left( t \right)}}$}}
\put (-112, 0) {{ \scriptsize {\it t}~/~s}}
\caption{Guaranteed-performance function for Lipschitz nonlinear cases.}
\end{center}\vspace{-1.8em}
\end{figure}

\begin{figure}[!htb]
\begin{center}
\scalebox{0.5}[0.5]{\includegraphics{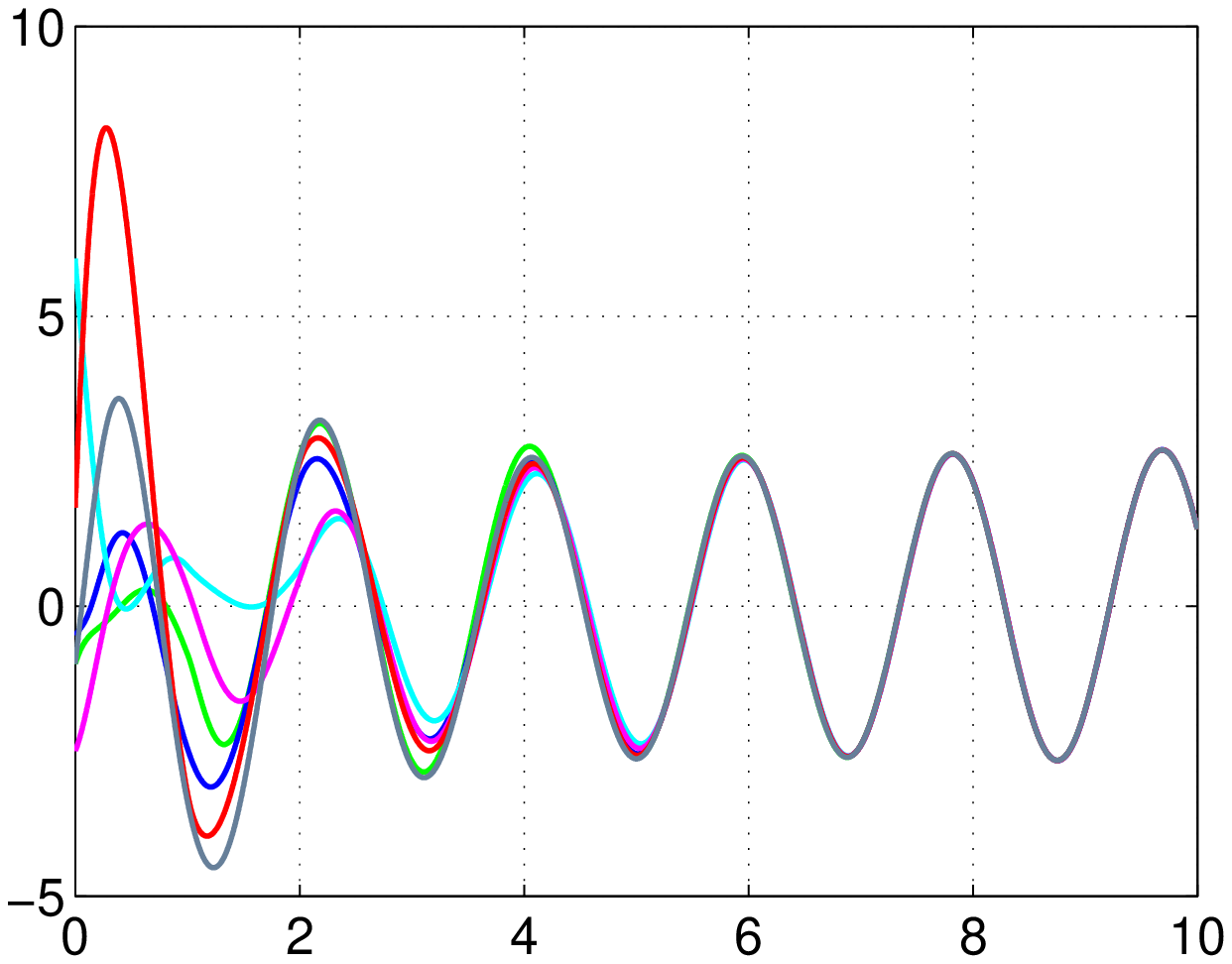}}
\put (-203, 42) {\rotatebox{90} {{\scriptsize $x_{i{\rm 1}}(t)~(i = 1,2, \cdots ,6)$}}}
\put (-110, 0) {{ \scriptsize {\it t}~/~s}}
\end{center}\vspace{-1em}
\end{figure}
\begin{figure}[!htb]
\begin{center}
\scalebox{0.5}[0.5]{\includegraphics{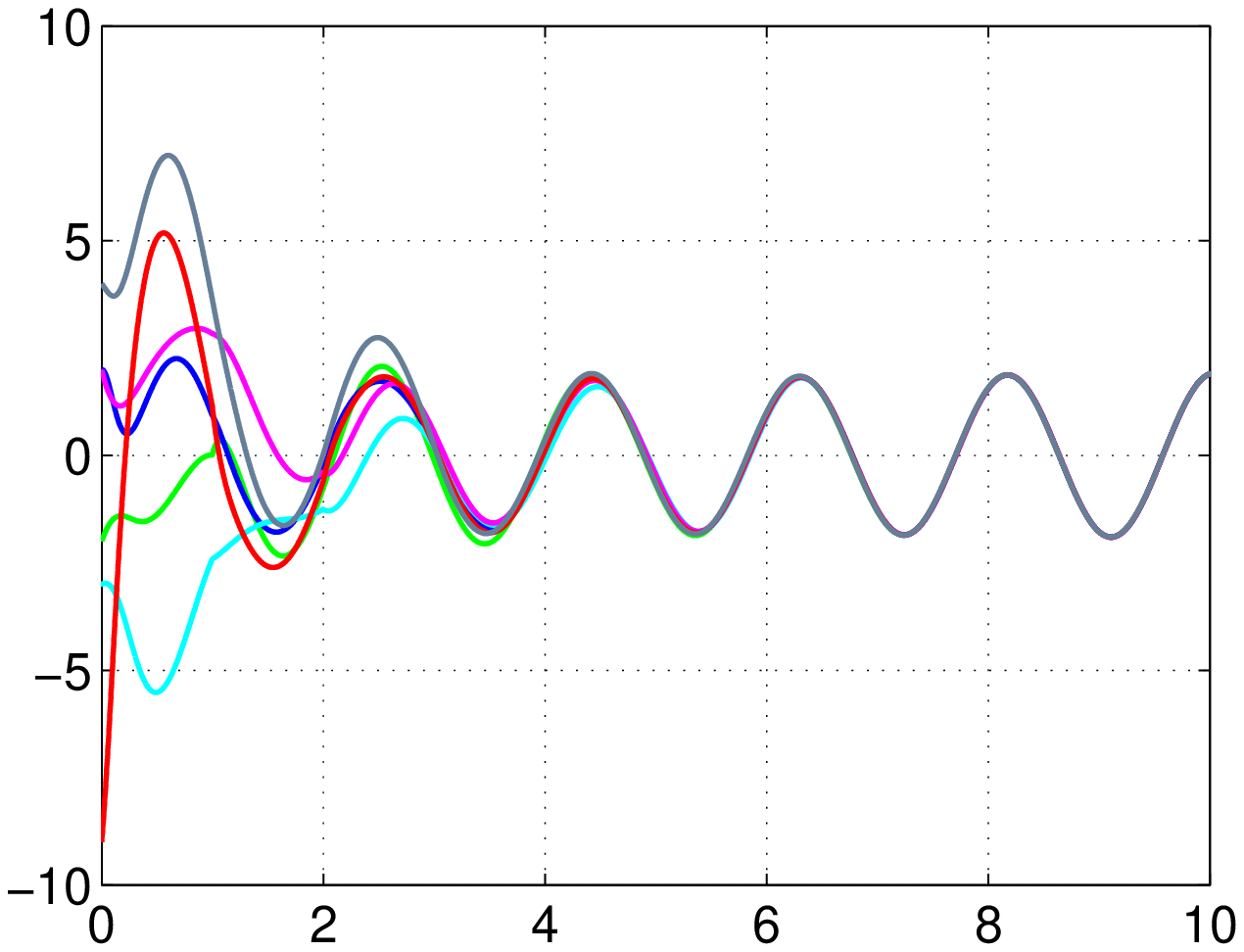}}
\put (-203, 42) {\rotatebox{90} {{\scriptsize $x_{i{\rm 2}}(t)~(i = 1,2, \cdots ,6)$}}}
\put (-110, 0) {{ \scriptsize {\it t}~/~s}}
\end{center}\vspace{-1em}
\end{figure}
\begin{figure}[!htb]
\begin{center}
\scalebox{0.5}[0.5]{\includegraphics{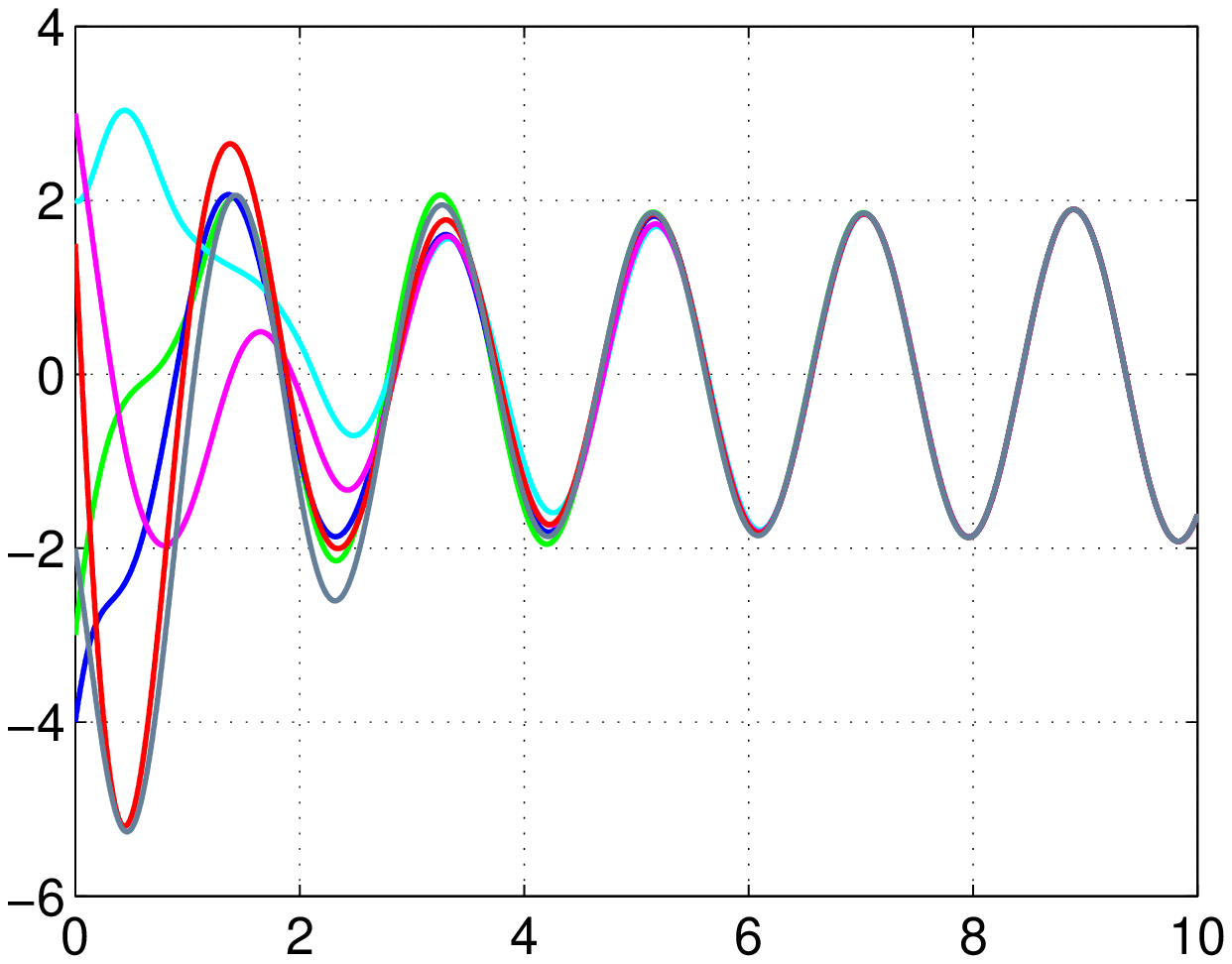}}
\put (-203, 42) {\rotatebox{90} {{\scriptsize $x_{i{\rm 3}}(t)~(i = 1,2, \cdots ,6)$}}}
\put (-110, 0) {{ \scriptsize {\it t}~/~s}}
\end{center}\vspace{-1em}
\end{figure}
\begin{figure}[!htb]
\begin{center}
\scalebox{0.5}[0.5]{\includegraphics{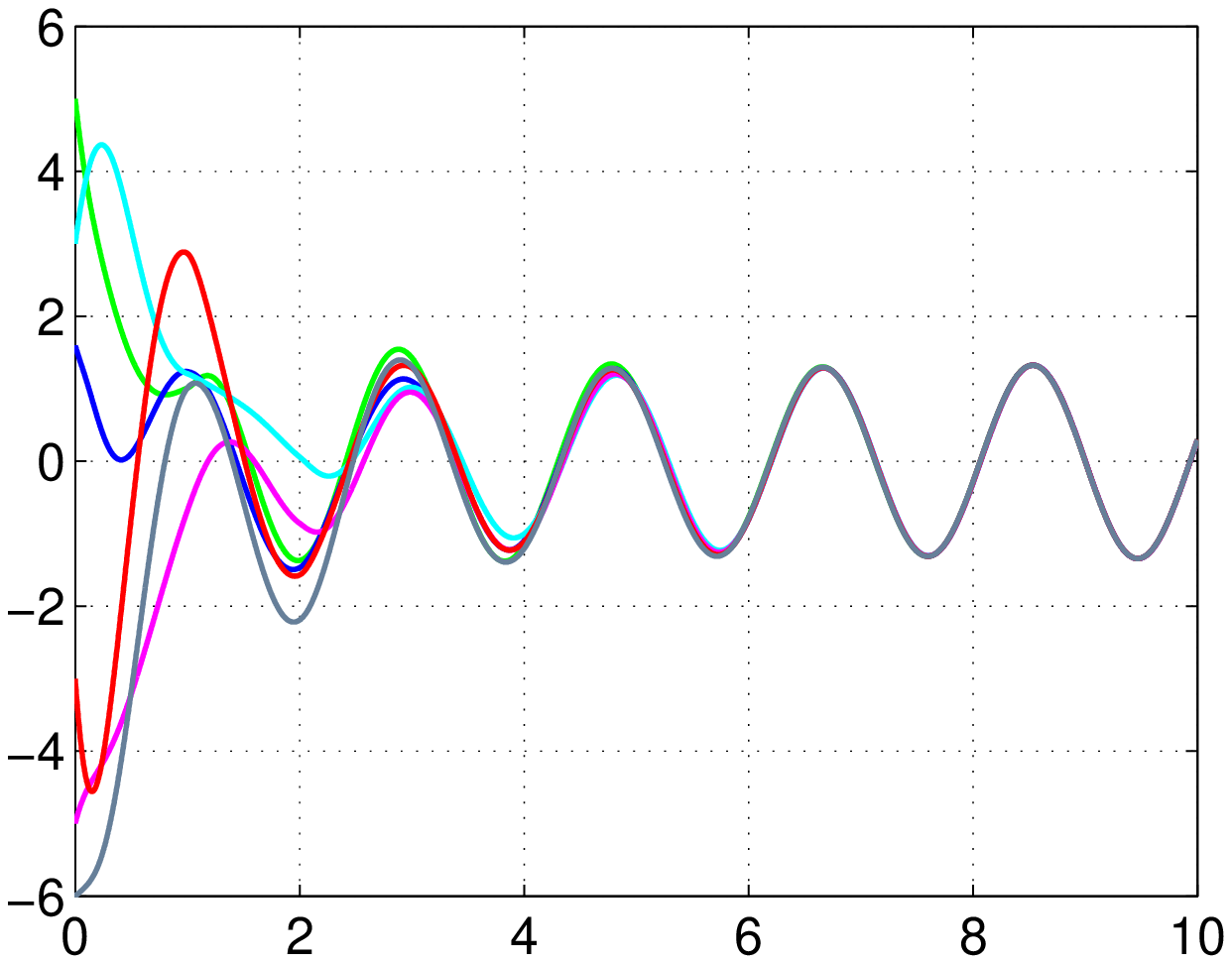}}
\put (-203, 42) {\rotatebox{90} {{\scriptsize $x_{i{\rm 4}}(t)~(i = 1,2, \cdots ,6)$}}}
\put (-110, 0) {{ \scriptsize {\it t}~/~s}}
\caption{State trajectories for Lipschitz nonlinear cases.}
\end{center}
\end{figure}

\end{document}